\documentclass[runningheads]{llncs}
\usepackage[T1]{fontenc}
\usepackage{amsmath} 
\usepackage{amssymb}  
\usepackage{graphicx,subfig}
\newcommand{\tp}{\normalfont \text{T}}

\newcommand{\Val}{\operatorname{Val}}
\newtheorem{assumption}{Assumption}

\newtheorem{corUpright}{Corollary}[section]
\let\oldcorUpright\corUpright
\renewcommand{\corUpright}{\oldcorUpright\normalfont}

\def\qedsymbol{\ensuremath{\Box}}      

\def\qed{\ifhmode\unskip\nobreak\fi\quad\qedsymbol}     
\def\frqed{\ifhmode\nobreak\hbox to5pt{\hfil}\nobreak%
	\hskip 0pt plus1fill\nobreak\fi\quad\qedsymbol\renewcommand{\qed}{}} 

\def\QEDsymbol{\vrule width.6em height.5em depth.1em\relax}
\def\frQED{\ifhmode\nobreak\hbox to5pt{\hfil}\nobreak%
	\hskip 0pt plus1fill\nobreak\fi\quad\QEDsymbol\renewcommand{\qed}{}} 
\def\QED{\ifhmode\unskip\nobreak\fi\quad\QEDsymbol}     

\begin{document}
\title{\texttt{FlipDyn} in Graphs: Resource Takeover Games in Graphs}
%
%
\author{Sandeep Banik\inst{1}\orcidID{0000-0001-9949-7173} \and
Shaunak D. Bopardikar\inst{2}\orcidID{0000-0002-0813-7867} \and
Naira Hovakimyan \inst{1}\orcidID{0000-0003-3850-1073
}}
\authorrunning{S. Banik et al.}
%
\institute{University of Illinois Urbana-Champaign,  Urbana, IL 61801-3633, USA\\
\email{baniksan@illinois.edu, nhovakim@illinois.edu} \and
Michigan State University, East Lansing, MI, 48823-24 \\
\email{shaunak@egr.msu.edu}}
\maketitle              
\begin{abstract}
We present \texttt{FlipDyn-G}, a dynamic game model extending the \texttt{FlipDyn} framework to a graph-based setting, where each node represents a dynamical system. This model captures the interactions between a defender and an adversary who strategically take over nodes in a graph to minimize (resp. maximize) a finite horizon additive cost. At any time, the \texttt{FlipDyn} state is represented as the current node, and each player can transition the \texttt{FlipDyn} state to a depending based on the connectivity from the current node. Such transitions are driven by the node dynamics, state, and node-dependent costs. This model results in a hybrid dynamical system where the discrete state (\texttt{FlipDyn} state) governs the continuous state evolution and the corresponding state cost. Our objective is to compute the Nash equilibrium of this finite horizon zero-sum game on a graph. Our contributions are two-fold. First, we model and characterize the \texttt{FlipDyn-G} game for general dynamical systems, along with the corresponding Nash equilibrium (NE) takeover strategies. Second, for scalar linear discrete-time dynamical systems with quadratic costs, we derive the NE takeover strategies and saddle-point values independent of the continuous state of the system. Additionally, for a finite state birth-death Markov chain (represented as a graph) under scalar linear dynamical systems, we derive analytical expressions for the NE takeover strategies and saddle-point values. We illustrate our findings through numerical studies involving epidemic models and linear dynamical systems with adversarial interactions.

\end{abstract}

\section{Introduction}
Cyber-Physical Systems (CPS) are integral to modern infrastructure, for a seamless integration of computational elements with physical processes to enable advanced functionalities in various domains. Some of the CPS encountered in real world include  smart grids, where computation coupled with sensor feedback ensures efficient energy distribution; autonomous vehicles, which rely on real-time data processing for navigation and safety; and industrial automation systems, which enhance productivity through precise control mechanisms and sensor feedback~\cite{lee2008cyber,shi2011survey}.

In the context of CPS, each node in a graph can be represented as a dynamical process, such as the generation and consumption of electricity in smart grids, the motion dynamics of autonomous vehicles, or the operational processes in industrial automation. These dynamical processes are interconnected through edges that represent the interactions and dependencies between them. For instance, in a smart grid, nodes may represent dynamic processes of energy generation and consumption at different substations, while edges denote the power flow between these substations~\cite{han2017mas}. Similarly, in autonomous vehicle networks, nodes could represent the dynamic driving processes of individual vehicles, with edges capturing the communication and coordination required for safe and efficient traffic flow~\cite{Yang2015stability,zheng2015stability,li2017dynamical}.

The use of graphs in modeling CPS is crucial for understanding the system's overall behavior and ensuring its robust operation. Graphs facilitate the visualization and analysis of how individual dynamic processes interconnect to form a larger, cohesive system. This interconnection highlights the importance of securing these nodes and their interactions to prevent disruptions that could compromise the entire system~\cite{bullo2009distributed,olfati2007consensus}.

A critical aspect of securing CPS involves understanding and mitigating the risks of stealthy takeovers, where an adversary covertly gains control of system components. The \texttt{FlipIT} game, introduced in~\cite{van2013flipit}, provides a framework for analyzing such stealthy takeovers. In \texttt{FlipIT}, both the attacker and defender can take control of a static resource, where the control changes can occur stealthily, without the immediate knowledge of the other party. This model captures the continuous and covert nature of security threats in CPS, emphasizing the need for persistent vigilance and strategic defense mechanisms.

The framework of \texttt{FlipIT} was extended to incorporate dynamical systems in  \texttt{FlipDyn}~\cite{FlipDyn_banik2022}, where both a defender and adversary aim to takeover a common resource modeled as a discrete-time dynamical system over a finite horizon. Building on the \texttt{FlipDyn} framework, this paper focuses on resource takeovers in graphs, where each node is represented as a resource with its own dynamics. The nodes are connected through edges that reflect the interactions between different CPS processes. The setup consists of two players, namely, a defender and an adversary who seek to repeatedly takeover the resources in the graph. This setup captures the strategic interactions between adversaries and defenders in a dynamic and interconnected environment, generalizing the \texttt{FlipDyn} setup to multiple \texttt{FlipDyn}  states.

Analyzing takeover games in this context involves understanding optimal strategies for both adversary and defender, considering various graph topologies and the specific characteristics of CPS. By leveraging game-theoretic models and the topology structure, this paper proposes robust defense mechanisms that enhance the resilience of CPS against takeover attacks. This approach is critical for ensuring the continued reliability and safety of essential infrastructures in the face of emerging cyber threats.

\subsection{Related Works}
The seminal work of stealthy takeover games, termed \texttt{FlipIT}~\cite{van2013flipit}, analyses a two player zero-sum game between a defender and an adversary attempting to takeover a static resource (computing device, virtual machine or a cloud service~\cite{bowers2012defending}). The work of \texttt{FlipIT} was generalized to the games of timing~\cite{johnson2015games}, where the actions of each player are dependent on the available exploitable vulnerability. Similarly, \texttt{FlipIT} has been extended to include time-based exponential discounting to the value of a protected resource~\cite{merlevede2019time}. In \texttt{FlipThem}~\cite{laszka2014flipthem}, the game of stealthy takeover was extended to multiple resources, which consists of an AND model, where all resources must be compromised for a takeover, and an OR model, where a player may compromise a single resource to takeover. The work in \texttt{FlipThem} was extended to i) a threshold-based version~\cite{leslie2015threshold}, which considered a finite number (threshold) of resources beyond which there exists no incentives to takeover, ii) multiple resource with constraints on the frequency of takeover actions~\cite{zhang2020defending}, and ii) heterogeneous resource costs and a learning-based method to determine player strategies~\cite{leslie2017multi}. Similar extensions include, Cheat-FlipIt model~\cite{yao2023cheat}, in which the opponent agent may feint to flip the resources first, and then control the resources after a finite delay. Such takeover strategies have also impacted the blockchain system~\cite{daian2020flash}, where arbitrage bots in decentralized exchanges engage in priority gas auctions to exploit against ordinary users. Beyond the domain of cybersecurity, the \texttt{FlipIT} model has been introduced in supervisory control and data acquisition (SCADA) to evaluate the impact of cyberattacks with insider assistance. The model of \texttt{FlipIT} has been extensively applied in system security~\cite{bowers2012defending}. These works primarily focused on resource takeovers within a static system, lacking consideration for the dynamic evolution of physical systems. In contrast, our work incorporates the dynamics of a physical system in the game of resource takeovers between an adversary and a defender, addressing the need for strategies that account for the continuous and evolving nature of CPS.


A finite-horizon zero-sum stochastic game has been studied to analyze probabilistic reachable sets for discrete-time stochastic hybrid systems~\cite{ding2013stochastic}, where both players simultaneously act on the system. Conversely, controllers have been synthesized~\cite{kontouras2014adversary} for intermittent switching between a defender and an adversary in discrete-time systems, with multi-dimensional control~\cite{kontouras2015covert} input and constraints. Such complete takeover of the system corresponds to covert misappropriation of a plant~\cite{smith2015covert}, where a feedback structure allows an attacker to take over control of the plant while remaining hidden from the supervisory system, extending such attacks in load frequency control (LFC) systems~\cite{mohan2020covert}. In contrast to previous research, our paper provides a feedback signal to infer who is in control and offers the ability to take control of the plant at any instant, balancing a trade-off between operational cost and performance.

The \texttt{FlipNet}\cite{saha2017flipnet} model extends the work of \texttt{FlipIT} to a graph, representing a networked system of multiple resources, where each player can take over the nodes of the graph. Network security in graphs is also viewed as advanced persistent threats (APT), which can be modeled as a zero-sum repeated game, where the states are represented as the set of compromised edges \cite{acquaviva2017optimal}. Similarly, APTs are modeled as multi-stage zero-sum network hardening games, where the adversary aims to determine the shortest path while the defender allocates resources in the graph to prevent the adversary from reaching the target node. Recently, dynamic information flow tracking has been proposed to detect APTs via a multistage game~\cite{moothedath2020game}. A similar model of an APT as a game played on an attack graph is explored in Cut-the-Rope~\cite{rass2019cut}, where the defender acts to cut the backdoor access of an adversary, demonstrating its efficacy on publicly documented attack graphs from the robotics domain~\cite{rass2023game}. The \texttt{FlipIT} model has also been used to study malware diffusion in epidemic models~\cite{miura2020modeling}. This work addresses \texttt{FlipIT} in a graph-based setting, where the defender and adversary repeatedly aim to take over the nodes of the graph. Unlike previous works, in this paper, the zero-sum game is played over a finite horizon with a discrete-time dynamical process on each node and time-varying costs.  

Our prior work which extends the \texttt{FlipIT} model to incorporate dynamical systems, termed \texttt{FlipDyn}~\cite{FlipDyn_banik2022}, presented a game of resource takeovers under a given control policy. The model of \texttt{FlipDyn} was extended to jointly solve the takeover and control policy~\cite{banik2023flipdyn}. In this paper, we extend the \texttt{FlipDyn} model to a finite horizon zero-sum game over a graph, where each node represents a dynamical system and the edges correspond to the interaction between these systems. The contributions of this work are two-fold: 

\begin{enumerate}
    \item \textbf{Takeover strategies over a graph with  discrete-time dynamical system on nodes}: We formulate a two-player zero-sum takeover game involving a defender and an adversary seeking to takeover the nodes of a graph, representing a discrete-time dynamical systems. The costs incurred by each player are contingent on the the current node of the graph. Assuming knowledge of the discrete-time dynamics, we establish the Nash equilibirum (NE) takeover strategies and saddle-point values. 
   \item \textbf{State-indepedent takeover strategies and saddle-point values for scalar/$1-$ dimensional systems}: For a linear discrete-time scalar dynamical system with quadratic takeover and state costs, we determine NE takeover policies independent of the continuous state of both players. Furthermore, for a  topology representing a finite state birth-death process, termed \emph{dual deter model}, we derive analytical expression of the NE takeover policies and saddle-point values.
\end{enumerate}

We illustrate our results on an epidemic model with no node dynamics and on a general graph with discrete-time dynamics on each node. 

This paper is structured as follows. Section~\ref{sec:Problem_Formulation} formally defines the \texttt{FlipDyn} problem in a graph setting with continuous state and node dependent costs. In Section~\ref{sec:FlipDyn-G_general}, we outline a solution methodology applicable to general discrete-time dynamical systems on nodes. Section~\ref{sec:FlipDyn_scalar_LQ} presents a solution for takeover policies for linear scalar discrete-time dynamical systems featuring quadratic costs, along with a topology dependent analytical solution and numerical examples. The paper concludes with a discussion on future directions in Section~\ref{sec:Conclusion}.

\section{Problem Formulation}\label{sec:Problem_Formulation}

Consider a directed multigraph $\mathcal{G}:=\{V,E, \phi \}$, where $V$ is the set of nodes with $|V| \in \mathbb{N}^{+}$, $E$ is the set of edges (paired nodes), and $\phi: E \to \{ \{\alpha, \beta\} | \alpha, \beta \in V^{2} \}$ is the  incidence function mapping every edge to an ordered pair of nodes, defining the connectivity of the graph. The term $e_{\alpha, \beta} \in E$ represents the edges connecting the node $\alpha \in V$ with the node $\beta \in V$, such that when $\alpha = \beta$, it represents a self-loop. We consider a single adversary, originating from any node of the graph $\mathcal{G}$. The adversary's goal is to reach nodes within the graph which induces maximum cost, while a defender's mission is to hinder the adversary's advances.

We consider the actions of the players and state evolution in discrete-time, with the variable $k$ denoting the current time step, which takes on values from the set $\mathcal{K} := \{ 1, 2, \ldots, L, L+1\}$. We represent the current node at time $k$ using a variable $\alpha_{k} \in V$, referred to as the \texttt{FlipDyn} state.  The adversary's action is denoted by the variable $\pi_{k}^{\text{a}} \in \epsilon(\alpha_{k})$, where the set $\epsilon(\alpha_{k})$ is defined as:
\begin{equation*}
    \epsilon(\alpha_{k}) := \bigcup j, j \in \{V | e_{\alpha_{k},j} \in E \}.
\end{equation*}
Here, $ \epsilon(\alpha_{k})$ represents the nodes the adversary can potentially target from the current node $\alpha_{k}$ at time $k$, with $j = \alpha_{k}$ indicating the choice to remain idle or stay in the same node. Similarly, the defender's action is denoted by $\pi_{k}^{\text{d}} \in \epsilon(\alpha_{k})$. Notice that the defender's action set is identical to that of the adversary's, to deter or prevent further escalation. The \texttt{FlipDyn} state update is based on both the action of the defender and adversary, given by:
\begin{equation}\label{eq:FlipDyn_update}
    \alpha_{k+1} = 
    \begin{cases}
        \pi^{d}_{k}, & \text{if } \pi_{k}^{\text{d}} = \pi_{k}^{\text{a}}, \\
        \pi_{k}^{\text{d}}, & \text{else if } \pi_{k}^{\text{d}} \in \{ \epsilon(\alpha_{k}) | \pi_{k}^{\text{a}} = \alpha_{k} \}, \\ 
        \pi_{k}^{\text{a}}, & \text{else if } \pi_{k}^{\text{a}} \in \{ \epsilon(\alpha_{k})| \pi_{k}^{\text{d}} = \alpha_{k}\}, \\
        \alpha_{k}, & \text{otherwise}.
    \end{cases}
\end{equation}
The \texttt{FlipDyn} update~\eqref{eq:FlipDyn_update} states that if the actions of both the defender and adversary are identical, then the \texttt{FlipDyn} state remains unchanged. However, if the defender opts to choose any node while the adversary remains idle, then the \texttt{FlipDyn} state transitions into the chosen node. Similarly, if the defender remains idle while the adversary chooses any node, then the \texttt{FlipDyn} transitions to the chosen node. The \texttt{FlipDyn} state transition can be compactly written as:
\begin{align}
    \begin{split}
        \label{eq:FD_u_compact}
        \alpha_{k+1} &= -\alpha_{k}{\mathbf{1}}_{\alpha_{k}}(\pi_{k}^{\text{a}}){\mathbf{1}}_{\alpha_{k}}(\pi_{k}^{\text{a}}) + \mathbf{1}_{\alpha_{k}}(\pi_{k}^{\text{a}})\pi_{k}^{\text{d}} + \mathbf{1}_{\alpha_{k}}(\pi_{k}^{\text{d}})\pi_{k}^{\text{a}}  + \bar{\mathbf{1}}_{\alpha_{k}}(\pi_{k}^{\text{a}})\bar{\mathbf{1}}_{\alpha_{k}}(\pi_{k}^{\text{a}})\pi_{k}^{d},
    \end{split}
\end{align}
where $\mathbf{1}_{\alpha_{k}}: 
\epsilon(\alpha_{k}) \to \{0,1\}$ is the indicator function, which maps to one if $\pi_{k}^{\text{d}} = \alpha_{k}$ or $\pi_{k}^{\text{a}} = \alpha_{k}$, and maps to zero, otherwise. The term  $\bar{\mathbf{1}}_{\alpha_{k}}$ is the one's complement of  ${\mathbf{1}}_{\alpha_{k}}$.
For illustrative purpose, consider the graph shown in Figure~\ref{fig:MSU_graph} with the \texttt{FlipDyn} state at time $k=1$ as  $\alpha_{\text{1}} = 1$. The \texttt{FlipDyn} state can transition to the node 2,~3 or remain in node 1 based on the update equation~\eqref{eq:FlipDyn_update}.
\begin{figure}[ht]
    \centering
    \includegraphics[width = 0.5\linewidth]{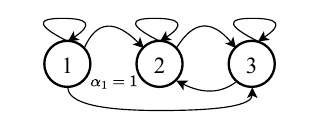}
    \caption{A directed multigraph consisting of 3 nodes. At time $k = 1$, the \texttt{FlipDyn} state is $\alpha_1 = 1$. The actions of both players are $\{ 1,2,3 \}$.}
    \label{fig:MSU_graph}
\end{figure}

In addition to the described graph environment, there is an underlying dynamical system whose continuous state at time $k$ is indicated by $x_{k} \in \mathcal{X} \subseteq  \mathbb{R}^{n}$, where $\mathcal{X}$ denotes the Euclidean state space. The state transition is dependent on the node $\alpha_{k+1}$ given by:
\begin{equation}
    \label{eq:CS_u}
    x_{k+1} = F_{k}^{\alpha_{k+1}}(x_{k}),
\end{equation}
where $F_{k}^{\alpha_{k+1}}: \mathcal{K} \to \mathcal{X}$ is the transition function for each $k \in \mathcal{K}$ and $\alpha_{k+1} \in V$ .

Our objective is to compute a strategy for both the players to transition the \texttt{FlipDyn} state to different nodes of the graph based on the dynamics~\eqref{eq:FD_u_compact},~\eqref{eq:CS_u}, takeover, and state costs. Given the initial state $x_{1}$ and node $\alpha_{1}$, we pose the node takeover problem as a zero-sum dynamic game governed by the \texttt{FlipDyn} update~\eqref{eq:FD_u_compact} and state dynamics~\eqref{eq:CS_u}, over a finite-time $L$, where the defender aims to minimize an additive cost given by: 
\begin{equation}\label{eq:obj_def_alpha}
	\begin{aligned}
		J(\alpha_1,  x_{1}, \{\pi^{\text{a}}_{\mathbf{L}}\}, \{\pi^{\text{d}}_{\mathbf{L}}\}) & = g_{L+1}^{\alpha_{L+1}}(x_{L+1}) + \sum_{t=1}^{L} g_t^{\alpha_t}(x_{t}) +   \bar{\mathbf{1}}_{\alpha_{t}}(\pi_{t}^{\text{d}})d_{t}^{\pi_{t}^{\text{d}}}(x_{t}) \\ & -\bar{\mathbf{1}}_{\alpha_{t}}( \pi_{t}^{\text{a}})a_{t}^{\pi_{t}^{\text{a}}}(x_{t}), 
	\end{aligned}
\end{equation}
where $g_{t}^{\alpha_t}(x_{t}): \mathcal{X} \rightarrow \mathbb{R}$ represents the cost for every \texttt{FlipDyn} state $\alpha_t \in V$, continuous state $x_{t}$ at time $t \in \mathcal{K}$, with $g_{L+1}^{\alpha_{L+1}}(x_{L+1}): \mathcal{X} \rightarrow \mathbb{R}$ representing the terminal cost for each $\alpha_{L+1} \in V$. The terms $d_{t}^{\pi_{t}^{\text{d}}}(x_{t}): \mathcal{X} \rightarrow \mathbb{R}$ and $a_{t}^{\pi_{t}^{\text{a}}}(x_{t}):  \mathcal{X} \rightarrow \mathbb{R}$ represent the instantaneous takeover costs of the defender and adversary, respectively, for each $t \in \mathcal{K}$ and action $\pi_{t}^{\text{d}},\pi_{t}^{\text{a}} \in \epsilon(\alpha_{t})$. 
The defender and adversary actions over the finite-horizon $L$ is given by the notations $\{\pi_{\mathbf{L}}^{\text{a}}\} := \{\pi_{1}^{\text{a}}, \dots, \pi_{L}^{\text{a}}\}$, and $\{\pi_{\mathbf{L}}^{\text{d}}\} := \{\pi_1^{\text{d}}, \dots, \pi_{L}^{\text{d}}\}$, respectively. In contrast, the adversary aims to maximize the cost function~\eqref{eq:obj_def_alpha} leading to a zero-sum dynamic game. This formulation characterizes the strategic interaction between the two players in the context of a node takeover problem in a graph environment, termed as \texttt{FlipDyn-G} game.

We seek to find Nash Equilibrium (NE) solutions of the game~\eqref{eq:obj_def_alpha}. To guarantee the existence of a pure or mixed NE takeover strategy, we expand the set of player policies to behavioral strategies --  probability distributions over the space of discrete actions at each discrete time~\cite{hespanha2017noncooperative}.
Specifically, let
\begin{equation}\label{eq:DF_bp}
	\mathbf{y}_{k}^{\alpha_{k}} := \bigcup_{j \in \epsilon(\alpha_{k})} y_{k,j}^{\alpha_{k}}, \sum_{j \in \epsilon(\alpha_{k})} y_{k,j} = 1, y_{k,j} \geq 0,  \text{ and }
\end{equation}
\begin{equation}\label{eq:AD_bp}
	\mathbf{z}_{k}^{\alpha_{k}} := \bigcup_{j \in \epsilon(\alpha_{k})} z_{k,j}^{\alpha_{k}}, \sum_{j \in \epsilon(\alpha_{k})} z_{k,j} = 1, z_{k,j} \geq 0
\end{equation}
be the behavioral strategies for the defender and adversary, respectively, at time instant $k$ for the \texttt{FlipDyn} state $\alpha_k$.
The takeover actions are
\[
\pi_{k}^{\text{d}} \sim \mathbf{y}_{k}^{\alpha_{k}}, \quad  \pi_{k}^{\text{a}} \sim \mathbf{z}_{k}^{\alpha_{k}},  
\]
for the defender and adversary at any time $k$ are sampled from the corresponding behavioral strategy. The behavioral strategies are $y_{k}^{\alpha_{k}}, z_{k}^{\alpha_{k}}  \in \Delta_{|\epsilon(\alpha_{k})|}$, where $\Delta_{|\epsilon(\alpha_{k})|}$ is the probability simplex in $|\epsilon(\alpha_{k})|$ dimensions. 
Over the finite horizon $L$, let $y_{\mathbf{L}} := \{\mathbf{y}_{1}^{\alpha_{1}}, \mathbf{y}_{2}^{\alpha_{2}}, \dots, \mathbf{y}_{L}^{\alpha_{L}}\} \in \Delta_{|\epsilon(\alpha_{1})|} \times \Delta_{|\epsilon(\alpha_{2})|} \times \dots \times \Delta_{|\epsilon(\alpha_{L})|}$ and $z_{\mathbf{L}} := \{\mathbf{z}_{1}^{\alpha_{1}}, \mathbf{z}_{2}^{\alpha_{2}}, \dots, \mathbf{z}_{L}^{\alpha_{L}}\} \in \Delta^{L}_{|\epsilon(\alpha_{1})|} \times \Delta^{L}_{|\epsilon(\alpha_{2})|} \times \dots \times \Delta^{L}_{|\epsilon(\alpha_{L})|}$ be the sequence of defender and adversary behavioral strategies. Thus, the expected outcome of the zero-sum game~\eqref{eq:obj_def_alpha} is given by:
\begin{equation}
    \label{eq:obj_def_E}
	J_{E}(x_1, \alpha_{1}, y_{\mathbf{L}}, z_{\mathbf{L}}) :=  \mathbb{E}[J( x_{1}, \alpha_{1}, \{\pi^{\text{a}}_{L}\}, \{\pi^{\text{d}}_{L}\})],
\end{equation}
where the expectation is computed with respect to the distributions $y_{\mathbf{L}}$ and $z_{\mathbf{L}}$. Specifically, we seek a saddle-point solution ($y_{\mathbf{L}}^{*}, z_{\mathbf{L}}^{*}$) in the space of behavioral strategies such that for any non-zero initial state $x_1 \in \mathcal{X}, \alpha_1 \in V$, we have:
\begin{equation*}
    \begin{aligned}
         J_{E}(x_1, \alpha_1, y_{\mathbf{L}}^{*}, z_{\mathbf{L}}) \leq J_E(x_1, \alpha_1, y_{\mathbf{L}}^{*}, z_{\mathbf{L}}^{*}) \leq J_{E}(x_1, \alpha_1, y_{\mathbf{L}}. z_{\mathbf{L}}^{*}).
    \end{aligned}
\end{equation*}
The \texttt{FlipDyn} game over a graph is completely defined by the expected cost~\eqref{eq:obj_def_E} and the space of player takeover strategies subject to the dynamics in~\eqref{eq:FD_u_compact} and \eqref{eq:CS_u}. In the next section, we derive the outcome of the \texttt{FlipDyn} game for each node in the graph for general systems.

\section{\texttt{FlipDyn-G} for General Problem}\label{sec:FlipDyn-G_general}
We will begin by defining the saddle-point value and deriving the NE takeover strategies of the \texttt{FlipDyn-G} game in any given node of the graph $\mathcal{G}$. 

\subsection{Saddle-point value of any node}
At time instant $k \in \mathcal{K}$, given a \texttt{FlipDyn} state $\alpha_{k}$, the saddle-point value consists of the instantaneous state cost and an additive cost-to-go based on the players takeover actions. The cost-to-go is determined via a cost-to-go matrix in each of the \texttt{FlipDyn} state $\alpha_{k}$, represented by $\Xi_{k+1}^{\alpha_{k}} \in \mathbb{R}^{|\epsilon(\alpha_{k})| \times |\epsilon(\alpha_{k})|}$. Let $V^{\alpha_{k}}_k(x, \Xi_{k+1}^{\alpha_{k}})$ be the saddle-point value at time instant $k$ with the continuous state $x$ and cost-to-go matrix, corresponding to the \texttt{FlipDyn} state of $\alpha_{k}$. Let us define a set of nodes connected to $\alpha_{k}$ as, $\{\alpha_{k}, j_{2}, j_{3}, \dots, j_{m(\alpha_{k})}\} \in \epsilon(\alpha_{k}),$ where $m(\alpha_{k}) = |\epsilon(\alpha_{k})|$. Such an ordered set of nodes will help us define the cost-to-go matrix. The entries of the cost-to-go matrix $\Xi^{\alpha_{k}}_{k+1}$ corresponding to each pair of takeover actions are given by:
\begin{equation}\label{eq:Cost_to_go_al0}
    \begin{aligned}
		& \begin{matrix} \hphantom{0000000} \text{$\alpha_{k}$} & \hphantom{00000000000} \text{$j_{2}$ } \hphantom{00000} & \dots & \hphantom{000000} \text{$j_{m(\alpha_{k})}$}\end{matrix} \\
		\begin{matrix} \text{$\alpha_{k}$} \\[3pt] \text{$j_{2}$} \\[3pt] \dots \\ \text{$j_{m(\alpha_{k})}$} \end{matrix} & \underbrace{\begin{bmatrix}
			v_{k+1}^{\alpha_{k}}(\alpha_{k},\alpha_{k}) &  \dots & \dots &  v_{k+1}^{j_{m(\alpha_{k})}}(\alpha_{k},j_{m(\alpha_{k})}) \\[3pt]
            v_{k+1}^{j_{2}}(j_{2},\alpha_{k}) &  v_{k+1}^{j_{2}}(j_{2},j_{2}) & \dots &  v_{k+1}^{\alpha_{k}}(j_{2},j_{m(\alpha_{k})}) \\ \dots & \dots & \dots & \dots \\
			v_{k+1}^{j_{m(\alpha_{k})}}(j_{m(\alpha_{k})},\alpha_{k}) &  v_{k+1}^{\alpha_{k}}(j_{m(\alpha_{k})},j_{2}) & \dots & v_{k+1}^{j_{m(\alpha_{k})}}(j_{m(\alpha_{k})},j_{m(\alpha_{k})})
		\end{bmatrix}}_{\Xi_{k+1}^{\alpha_{k}}}
	\end{aligned},
\end{equation}
where $v_{k+1}^{\alpha_{k+1}} (\pi_{k}^{\text{d}},\pi_{k}^{\text{a}})$ corresponds to the cost-to-go value of a \texttt{FlipDyn} state $\alpha_{k+1} \in V$, defined as:
\begin{equation*}
    \begin{split}
        v_{k+1}^{\alpha_{k+1}}(\pi_{k}^{\text{d}},\pi_{k}^{\text{a}}) &:= V_{k+1}^{\alpha_{k+1}}(F^{\alpha_{k+1}}_{k}(x), \Xi_{k+2}^{\alpha_{k+2}}) + \bar{\mathbf{1}}_{\alpha_{k}}(\pi_{k}^{\text{d}})d_{k}^{\pi_{k}^{\text{d}}}(x_{k}) \\ & -\bar{\mathbf{1}}_{\alpha_{k}}( \pi_{k}^{\text{a}})a_{k}^{\pi_{k}^{\text{a}}}(x_{k}).
    \end{split}
\end{equation*}
The diagonal terms in~\eqref{eq:Cost_to_go_al0} correspond to the saddle-point value of the \texttt{FlipDyn} states under identical defender and adversary actions. Notice, only under the action of $\pi_{k}^{\text{d}} = \pi_{k}^{\text{a}} = \alpha_{k}$ the takeover costs for both players are zero. The first row of $\Xi_{k+1}^{\alpha_{k}}$ corresponds to the saddle-point values of \texttt{FlipDyn} states chosen by the adversary, when the defender remains idle. Similarly, the first column corresponds to the saddle-point value of the \texttt{FlipDyn} states chosen by the defender under an idle adversary action. The remaining entries of $\Xi_{k+1}^{\alpha_{k}}$ correspond to the saddle-point value of the \texttt{FlipDyn} state $\alpha_{k}$ with the corresponding takeover costs. The entries of the cost-to-go matrix are constructed using the \texttt{FlipDyn} dynamics~\eqref{eq:FD_u_compact} and continuous state dynamics~\eqref{eq:CS_u}. Thus, at time $k$ for a given state $x$ and $\alpha_k$, the saddle-point value satisfies
\begin{equation}
    \label{eq:V_k_saddle_point}
	V^{\alpha_{k}}_k(x, \Xi_{k+1}^{\alpha_{k}}) = g_k^{\alpha_{k}}(x) + \Val(\Xi^{\alpha_{k}}_{k+1}), 
\end{equation}
where $\Val(X_{k+1}^{\alpha_{k}}):= \min_{y_{k}^{\alpha_{k}}} \max_{z_{k}^{\alpha_{k}}} y_{k}^{{\alpha_{k}}^{\tp}}X_{k+1}z_{k}^{\alpha_{k}}$, represents the (mixed) saddle-point value of the zero-sum matrix $X_{k+1}$ for the \texttt{FlipDyn} state $\alpha_{k}$, and $\Xi^{\alpha_{k}}_{k+1} \in \mathbb{R}^{|\epsilon(\alpha_{k})| \times |\epsilon(\alpha_{k})|}$ is the cost-to-go zero-sum matrix. The defender's and adversary's action results in either an entry within $\Xi^{\alpha_{k}}_{k+1}$ (if the matrix has a saddle point in pure strategies) or in the expected sense, resulting in a cost-to-go from state $x$ at time $k$.

 With the saddle-point values established in each of the \texttt{FlipDyn} states $\alpha_{k} \in V$, next, we will characterize the NE takeover strategies and the saddle-point values for the entire time horizon $L$.

\subsection{NE takeover strategies of the \texttt{FlipDyn-G} game}
To characterize the saddle-point value of the game, we restrict the state and takeover costs to a particular domain, stated in the following mild assumption.
\begin{assumption}\label{ast:general_costs}
    [Non-negative costs] For any time instant $k \in \mathcal{K}$, the state and takeover costs $g_k^{\alpha}(x), d_{k}^{\alpha}(x), a_{k}^{\alpha}(x),$ for all $x \in \mathcal{X},$ and $\alpha \in V$ are non-negative $(\mathbb{R}_{\geq 0})$. 
\end{assumption}

Assumption~\ref{ast:general_costs} enables us to compare the entries of the cost-to-go matrix without changes in the sign of the costs, thereby, characterizing the strategies of the players (pure or mixed strategies). 
Under Assumption~\ref{ast:general_costs}, we derive the following result to compute a  recursive saddle-point value for the horizon length $L$ and the corresponding NE takeover strategies for both the players in every node of the graph environment.

\begin{lemma}\label{lem:SPV_iter_gen_FlipDynG}
    Under Assumption~\ref{ast:general_costs}, the saddle-point value of the \texttt{FlipDyn-G} game~\eqref{eq:obj_def_E} at any time $k \in \mathcal{K}$, subject to the \texttt{FlipDyn} dynamics~\eqref{eq:FD_u_compact} and continuous state dynamics~\eqref{eq:CS_u} is given by:
    \begin{align}
    \label{eq:SPV_gen_FlipG}
        V_{k}^{\alpha_{k}*}(x,\Xi_{k+1}^{\alpha_{k}}) = g_{k}^{\alpha_{k}} + y_{k}^{\alpha_{k}*^{\tp}}\Xi_{k+1}^{\alpha_{k}}z_{k}^{\alpha_{k}*}, 
    \end{align}
    where $y_{k}^{\alpha_{k}*}$ and $z_{k}^{\alpha_{k}*}$ correspond to NE takeover policies obtained upon solving the zero-sum matrix defined by $\Xi_{k+1}^{\alpha_{k}}$ as a linear program~\cite{hespanha2017noncooperative}.
    The boundary condition of the saddle-point value recursion~\eqref{eq:SPV_gen_FlipG} at $k = L$ is given by:
    \begin{equation}\label{eq:b_cond_val_gen}
        \Xi_{L+2}^{\alpha_{L+1}} := \mathbf{0}_{m(\alpha_{L+1}) \times m(\alpha_{L+1})}, \forall \alpha_{L+1} \in V.    
    \end{equation}
    \frqed
\end{lemma}

We skip the proof of the Lemma~\ref{lem:SPV_iter_gen_FlipDynG} as it involves simple substitutions and the use of recursive optimality. For a finite cardinality of the state space $\mathcal{X}$, \texttt{FlipDyn} states $V$, and a finite horizon $L$, Lemma~\ref{lem:SPV_iter_gen_FlipDynG} yields an exact (behavioral) saddle-point value of the \texttt{FlipDyn-G} game~\eqref{eq:obj_def_E}. However, the computational and storage complexities scale undesirably with the cardinality of $\mathcal{X}$, especially in continuous state spaces. For this purpose, in the next section, we will provide a parametric form of the saddle-point value especially in the case of scalar linear dynamics with quadratic costs.

\section{\texttt{FlipDyn-G} for scalar LQ Problems}\label{sec:FlipDyn_scalar_LQ}
To render a tractable solution for continuous state of the \texttt{FlipDyn-G} game, we restrict ourselves to scalar linear discrete-time dynamical system with quadratic costs (LQ problem). The discrete-time dynamics of a linear system at time instant $k \in \mathcal{K}$ in the \texttt{FlipDyn} state $\alpha_{k+1}$ is given by:
\begin{equation}
    \label{eq:scalar_alpha_dynamics}
    \begin{aligned}
        x_{k+1} & = F_{k}^{\alpha_{k+1}}(x_k) := f_k^{\alpha_{k+1}}x_k,
    \end{aligned}
\end{equation}
where $f_{k}^{\alpha_{k+1}} \in \mathbb{R}$ denotes the state transition scalar coefficient.
The stage and takeover costs are assumed to be quadratic for each player and given by:
\begin{gather}\label{eq:ST_MV_cost_Q}
        g_k(x,\alpha_k) = x^{2}\mathbf{g}_k^{\alpha_k}, \quad d_{k}^{\alpha_k}(x) = x^{2}\mathbf{d}_{k}^{\alpha_k}, \quad {a}_{k}^{\alpha_k}(x) = x^{2}\mathbf{a}_{k}^{\alpha_k}, 
\end{gather}
where $\mathbf{g}_k^{\alpha_k} \in \mathbb{R}, \mathbf{a}_{k}^{\alpha_k} \in \mathbb{R}, \mathbf{d}_{k}^{\alpha_k} \in \mathbb{R}$ are non-negative $(\mathbb{R}_{\geq 0})$ under Assumption~\ref{ast:general_costs}.

Under Assumption~\ref{ast:general_costs} for scalar dynamical systems of the form~\eqref{eq:scalar_alpha_dynamics}, we postulate a parametric form for the saddle-point value for each \texttt{FlipDyn} state $\alpha \in V$ of the form:
\begin{equation}
    \label{eq:para_form_al_QC}
    \begin{aligned}
        V_{k}^{\alpha_{k}}(x, \Xi_{k+1}^{\alpha_{k}}) \Rightarrow V_{k}^{\alpha_{k}}(x) := \mathbf{p}^{\alpha_{k}}_{k}x^{2}, \ \forall \alpha_{k} \in V, \ k \in \mathcal{K},
    \end{aligned}
\end{equation}
where $\mathbf{p}^{\alpha_{k}}_{k} \in \mathbb{R}_{\geq 0}$ corresponds to a non-negative coefficient for each of the \texttt{FlipDyn} states. Under the scalar linear dynamical system~\eqref{eq:scalar_alpha_dynamics}, takeover costs~\eqref{eq:ST_MV_cost_Q} and the parameteric form~\eqref{eq:para_form_al_QC}, the cost-to-go matrix $\hat{\Xi}_{k+1}^{\alpha_{k}}$ can be re-expressed as:
\begin{equation}\label{eq:Cost_to_go_scalar}
    \begin{aligned}
		& \begin{matrix} \hphantom{0000000} \text{$\alpha_{k}$} & \hphantom{00000000000} \text{$j_{2}$ } \hphantom{00000} & \dots & \hphantom{000000} \text{$j_{m(\alpha_{k})}$}\end{matrix} \\
		\begin{matrix} \text{$\alpha_{k}$} \\[3pt] \text{$j_{2}$} \\[3pt] \dots \\ \text{$j_{m(\alpha_{k})}$} \end{matrix} & \underbrace{\begin{bmatrix}
			\mathbf{v}_{k+1}^{\alpha_{k}}(\alpha_{k},\alpha_{k}) &  \dots & \dots &  \mathbf{v}_{k+1}^{j_{m(\alpha_{k})}}(\alpha_{k},j_{m(\alpha_{k})}) \\[3pt]
            \mathbf{v}_{k+1}^{j_{2}}(j_{2},\alpha_{k}) &  \mathbf{v}_{k+1}^{j_{2}}(j_{2},j_{2}) & \dots &  \mathbf{v}_{k+1}^{\alpha_{k}}(j_{2},j_{m(\alpha_{k})}) \\ \dots & \dots & \dots & \dots \\
			\mathbf{v}_{k+1}^{j_{m(\alpha_{k})}}(j_{m(\alpha_{k})},\alpha_{k}) &  \mathbf{v}_{k+1}^{\alpha_{k}}(j_{m(\alpha_{k})},j_{2}) & \dots & \mathbf{v}_{k+1}^{j_{m(\alpha_{k})}}(j_{m(\alpha_{k})},j_{m(\alpha_{k})})
		\end{bmatrix}}_{\hat{\Xi}_{k+1}^{\alpha_{k}}}
	\end{aligned},
\end{equation}
where $\mathbf{v}_{k+1}^{\alpha_{k}} (u,w)$ corresponds to the cost-to-go  term of a \texttt{FlipDyn} state independent of the term $x^{2}$, defined as:
\begin{equation*}
    \begin{split}
        \mathbf{v}_{k+1}^{\alpha_{k+1}}(\pi_{k}^{\text{d}},\pi_{k}^{\text{a}}) &:= (f^{\alpha_{k+1}}_{k})^{2}\mathbf{p}_{k+1}^{\alpha_{k+1}} + \bar{\mathbf{1}}_{\alpha_{k}}(\pi_{k}^{\text{d}})\mathbf{d}_{k}^{\pi_{k}^{\text{d}}} -\bar{\mathbf{1}}_{\alpha_{k}}( \pi_{k}^{\text{a}})\mathbf{a}_{k}^{\pi_{k}^{\text{a}}}.
    \end{split}
\end{equation*}
Notice the cost-to-go entries consists of the system transition coefficients and takeover costs, while factoring out the term $x^{2}$. 
Building on Lemma~\ref{lem:SPV_iter_gen_FlipDynG}, we present the following result, which provides the NE takeover policies of both players, and outlines the saddle-point value update of $\mathbf{p}_k^{\alpha_{k}}$ for any \texttt{FlipDyn} state.

\begin{lemma}\label{lem:NE_Val_scalar}
    Under Assumption~\ref{ast:general_costs}, at any time $k \in \mathcal{K}$, the saddle-point value parameter of the \texttt{FlipDyn-G} game~\eqref{eq:obj_def_E} for quadratic state and takeover costs~\eqref{eq:ST_MV_cost_Q}, subject to the \texttt{FlipDyn} dynamics~\eqref{eq:FD_u_compact} and scalar state dynamics~\eqref{eq:scalar_alpha_dynamics}, is given by:
    \begin{align}\label{eq:SPV_scalar_FlipG}
        \mathbf{p}_{k}^{\alpha_{k}*} = g_{k}^{\alpha_{k}} + y_{k}^{\alpha_{k}*^{\tp}}\hat{\Xi}_{k+1}^{\alpha_{k}}z_{k}^{\alpha_{k}*},
    \end{align}
    where $y_{k}^{\alpha_{k}*}$ and $z_{k}^{\alpha_{k}*}$ correspond to NE takeover policies obtained upon solving the zero-sum matrix $\hat{\Xi}_{k+1}^{\alpha_{k}}$ as a linear program~\cite{hespanha2017noncooperative}.
    The boundary condition of the saddle-point value recursion~\eqref{eq:SPV_gen_FlipG} at $k = L$ is given by:
    \begin{align}\label{eq:b_cond_scalar_gen}
        \hat{\Xi}_{L+2}^{\alpha_{L+1}} := \mathbf{0}_{m(\alpha_{L+1}) \times m(\alpha_{L+1})}, \forall \alpha_{L+1} \in V.    
    \end{align}
    \frqed
\end{lemma}

Substituting the scalar state dynamics~\eqref{eq:scalar_alpha_dynamics} along with state and takeover costs~\eqref{eq:ST_MV_cost_Q} yields the NE strategies and saddle-point value parameters~\eqref{eq:SPV_scalar_FlipG}.
We skip the proof of Lemma~\ref{lem:NE_Val_scalar} for brevity. Lemma~\ref{lem:NE_Val_scalar} presents a complete solution for the \texttt{FlipDyn-G}~\eqref{eq:obj_def_E} game with NE takeover strategies independent of state of the scalar dynamical system. 
\begin{remark}
   The presented analysis also extends to models that are affine in the sate with a bias term. Such a model will require the saddle-point value parameterization of the form:
\[
V_{k}^{\alpha_{k}}(x) := \mathbf{p}^{\alpha_{k}}_{k}x^{2} + \mathbf{q}^{\alpha_{k}}_{k}, \ \forall \alpha_{k} \in V, \ k \in \mathcal{K},
\]
where $\mathbf{p}^{\alpha_{k}}_{k}, \mathbf{q}^{\alpha_{k}}_{k} \in \mathbb{R}_{\geq 0}$ corresponds to a non-negative coefficient for each of the \texttt{FlipDyn} states.
\end{remark}

In the following subsection, we will derive closed-form expressions of the \texttt{FlipDyn-G} game for a special graph structure and show how the structure  represents the original \texttt{FlipDyn} game~\cite{FlipDyn_banik2022}.

\subsection{Dual Deter \texttt{FlipDyn-G} game}
We look at a special case of the graph environment, which consists of a start and an end node connecting only one other node, while the remaining nodes connect two different nodes, termed as the \emph{dual deter} model. This model can be viewed as finite state Markov chain birth-death process~\cite{li1989overload}. We assume that such a dual deter model has an ordered set of nodes from node $0$ to $N$, which implies that the dual deter model has a total of $|V| = N+1$ nodes, illustrated in Figure~\ref{fig:Dual_in_out_graph}. 
\begin{figure}[ht]
    \centering
    \includegraphics[width = 0.5\linewidth]{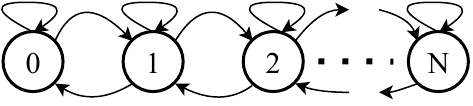}
    \caption{A graph consisting of $N$ nodes .}
    \label{fig:Dual_in_out_graph}
\end{figure}
 
 A key difference compared to the general graph model lies in the action space of the defender and adversary. At any node $\alpha_{k} \in \{1,2,\dots,N-1\}$, the action space of the adversary is $\pi_{k}^{\text{a}} := \{\alpha_{k}, \overline{\alpha}\}, \ \overline{\alpha} \in \{V | \overline{\alpha} > \alpha_{k} \}$, and of the defender is $\pi_{k}^{\text{d}} := \{\alpha_{k}, \underline{\alpha}\}, \ \underline{\alpha} \in \{V | \underline{\alpha} < \alpha_{k}\}$. The action space of both the defender and adversary in the start and end node $\alpha_{k} = \{0,N\}$ are given as $\pi_{k}^{\text{d}} := \{\alpha_{k},\tau\}$ and $\pi_{k}^{\text{a}} := \{\alpha_{k},\tau\}$, where $\tau$ represents a takeover action in the node $\alpha_{k}$, preventing transition to other nodes.  
Such an action space and model represents the defender deterring an adversary from escalating through the graph. The \texttt{FlipDyn} state updates  in such a dual deter model as follows:
\begin{equation}
    \label{eq:FD_dual_deter}
    \alpha_{k+1} = 
    \begin{cases}
        \alpha_{k}, & \text{if } \pi_{k}^{\text{d}} = \pi_{k}^{\text{a}} | \alpha_{k} = \{0,N\}, \\
        1, & \text{else if } \pi_{k}^{\text{a}} = \tau | \pi_{k}^{\text{d}} = 0, \alpha_{k} = 0, \\
        N-1, & \text{else if } \pi_{k}^{\text{d}} = \tau | \pi_{k}^{\text{a}} = N, \alpha_{k} = N, \\
        \pi^{d}_{k}, & \text{else if } \pi_{k}^{\text{d}} = \pi_{k}^{\text{a}}, \\
        \underline{\alpha}, & \text{else if } \pi_{k}^{\text{d}} = \underline{\alpha} | \pi_{k}^{\text{a}} = \alpha_{k} , \\ 
        \overline{\alpha}, & \text{else if} \pi_{k}^{\text{a}} = \overline{\alpha}| \pi_{k}^{\text{d}} = \alpha_{k}\}, \\
        \alpha_{k}, & \text{otherwise}.
    \end{cases}
\end{equation}

We characterize the NE strategies and saddle-point values of the dual deter model under the assumption of a scalar linear dynamical system~\eqref{eq:scalar_alpha_dynamics} and quadratic costs~\eqref{eq:ST_MV_cost_Q} with a parameterized saddle-point value~\eqref{eq:para_form_al_QC}. Such an action space leads to a reduced dimension of the cost-to-go matrix independent of the state term $x^{2}$ at any node $\alpha_{k} \in \{1,2,\dots,N-1\}$, given by:
\begin{equation}\label{eq:Cost_to_go_dual}
    \begin{aligned}
		& \begin{matrix} \hphantom{0000000} \alpha_{k} & \hphantom{00000000000} \overline{\alpha} \hphantom{00000} \end{matrix} \\
		\begin{matrix} \alpha_{k} \\[3pt] \underline{\alpha} \end{matrix} & \begin{bmatrix}
			\left(f^{\alpha_{k}}_{k}\right)^2 \mathbf{p}_{k+1}^{\alpha_{k}} & \left(f^{\overline{\alpha}}_{k}\right)^2 \mathbf{p}_{k+1}^{\overline{\alpha}} - \mathbf{a}_{k}^{\alpha_{k}} \\[3pt]
            \left(f^{\underline{\alpha}}_{k}\right)^2 \mathbf{p}_{k+1}^{\underline{\alpha}} + \mathbf{d}_{k}^{\alpha_{k}} & \left(f^{\alpha_{k}}_{k}\right)^2 \mathbf{p}_{k+1}^{\alpha_{k}} + \mathbf{d}_{k}^{\alpha_{k}} - \mathbf{a}_{k}^{\alpha_{k}}
		\end{bmatrix}
	\end{aligned}.
\end{equation}
Similarly, the cost-to-go matrix for the start node $\alpha_{k} = 0$ independent of the state term $x^{2}$ is given by:
\begin{equation}\label{eq:Cost_to_go_start}
    \begin{aligned}
		& \begin{matrix} \hphantom{0000000} 0 & \hphantom{00000000000} \tau \hphantom{00000} \end{matrix} \\
		\begin{matrix} 0 \\[3pt] \tau \end{matrix} & \begin{bmatrix}
			\left(f^{0}_{k}\right)^2 \mathbf{p}_{k+1}^{0} & \left(f^{1}_{k}\right)^2 \mathbf{p}_{k+1}^{1} - \mathbf{a}_{k}^{0} \\[3pt]
            \left(f^{0}_{k}\right)^2 \mathbf{p}_{k+1}^{0} + \mathbf{d}_{k}^{0} & \left(f^{0}_{k}\right)^2 \mathbf{p}_{k+1}^{0} + \mathbf{d}_{k}^{0} - \mathbf{a}_{k}^{0}
		\end{bmatrix}
	\end{aligned},
\end{equation}
whereas for the end node $\alpha_{k} = N$, we have:
\begin{equation}\label{eq:Cost_to_go_end}
    \begin{aligned}
		& \begin{matrix} \hphantom{0000000} N & \hphantom{0000000000000000000} \tau \hphantom{00000} \end{matrix} \\
		\begin{matrix} N \\[3pt] \tau \end{matrix} & \begin{bmatrix}
			\left(f^{N}_{k}\right)^2 \mathbf{p}_{k+1}^{N} & \left(f^{N}_{k}\right)^2 \mathbf{p}_{k+1}^{N} - \mathbf{a}_{k}^{N} \\[3pt]
            \left(f^{N-1}_{k}\right)^2 \mathbf{p}_{k+1}^{N-1} + \mathbf{d}_{k}^{N} & \left(f^{N}_{k}\right)^2 \mathbf{p}_{k+1}^{N} + \mathbf{d}_{k}^{N} - \mathbf{a}_{k}^{N}
		\end{bmatrix}
	\end{aligned}.
\end{equation}
The transition of the nodes in~\eqref{eq:Cost_to_go_dual} follows from the \texttt{FlipDyn} dynamics~\eqref{eq:FD_u_compact}. Next, we present the NE takeover in both pure and mixed strategies of both the players along with the saddle-point value parameter $\mathbf{p}_{k}^{\alpha_{k}}$ for every node in the dual deter model.

\begin{theorem}\label{th:NE_Val_scalar_dual_deter}
      The unique NE takeover strategies of the \texttt{FlipDyn-G} game~\eqref{eq:obj_def_E} at any time $k \in \mathcal{K}$ for quadratic state and takeover costs~\eqref{eq:ST_MV_cost_Q}, subject to the \texttt{FlipDyn} dynamics~\eqref{eq:FD_dual_deter} and scalar state dynamics~\eqref{eq:scalar_alpha_dynamics} are given by:
      
      \underline{Case i) - $\alpha_{k} = 0$}
      \begin{align}
         \begin{split}\label{eq:def_TP_alpha_0}
                 y^{0*}_{k}  = \begin{cases}
                    \begin{bmatrix} \dfrac{\mathbf{a}^{0}_{k}}{\hat{\mathbf{p}}_{k+1}} & 1 - \dfrac{\mathbf{a}^{0}_{k}}{\hat{\mathbf{p}}_{k+1}} 
                \end{bmatrix}^{\tp}, & \text{if } \ 
                    \begin{aligned}
                        & \hat{\mathbf{p}}_{k+1} > \mathbf{a}^{0}_{k},  \hat{\mathbf{p}}_{k+1} > \mathbf{d}^{0}_{k},
                    \end{aligned}
                    \\[15pt]
                    \begin{bmatrix} \hphantom{00} 1 \hphantom{0} & \hphantom{0000} 0 \hphantom{000}
                \end{bmatrix}^{\tp}, & \text{otherwise,}
                \end{cases} 
        \end{split}\\
        \begin{split}\label{eq:adv_TP_alpha_0}
                 z^{0*}_{k}  = \begin{cases}
                    \begin{bmatrix}  1 - \dfrac{\mathbf{d}^{0}_{k}}{\hat{\mathbf{p}}_{k+1}} & \dfrac{\mathbf{d}^{0}_{k}}{\hat{\mathbf{p}}_{k+1}}
                \end{bmatrix}^{\tp}, & \text{if } \ 
                    \begin{aligned}
                        & \hat{\mathbf{p}}_{k+1} > \mathbf{a}^{0}_{k},   \hat{\mathbf{p}}_{k+1} > \mathbf{d}^{0}_{k},
                    \end{aligned}
                    \\[15pt]
                    \begin{bmatrix} \hphantom{00} 0 \hphantom{00} & \hphantom{000} 1 \hphantom{000}
                \end{bmatrix}^{\tp}, & \text{if } \ 
                    \begin{aligned}
                        & \hat{\mathbf{p}}_{k+1} > \mathbf{a}^{0}_{k},   \hat{\mathbf{p}}_{k+1} \leq \mathbf{d}^{0}_{k},
                    \end{aligned}
                    \\[15pt]
                \begin{bmatrix} \hphantom{00} 1 \hphantom{00} & \hphantom{000} 0 \hphantom{000}
                \end{bmatrix}^{\tp}, & \text{otherwise,}
                \end{cases} 
        \end{split}
        \end{align}
       
        and the saddle-point value parameter satisfies:
        \begin{align}\label{eq:Val_scalar_alpha_0}
            \mathbf{p}_{k}^{0} = 
            \begin{cases}
                \begin{aligned}
    				 \mathbf{g}_{k}^{0} + & (f_{k}^{0})^{2}\mathbf{p}_{k+1}^{0} + \mathbf{d}_{k}^{0} - \dfrac{\mathbf{a}^{0}_{k}\mathbf{d}^{0}_{k}}{\hat{\mathbf{p}}_{k+1}},
                \end{aligned} & \text{if } \ 
                    \begin{aligned}
                        & \hat{\mathbf{p}}_{k+1} > \mathbf{a}^{0}_{k},   \hat{\mathbf{p}}_{k+1} > \mathbf{d}^{0}_{k},
                    \end{aligned}
                    \\[15pt]
                \begin{aligned}
    				\mathbf{g}_{k}^{0} + & (f_{k}^{1})^{2}\mathbf{p}_{k+1}^{1} - \mathbf{a}_{k}^{0} ,
                \end{aligned} & \text{if } \ 
                    \begin{aligned}
                        & \hat{\mathbf{p}}_{k+1} > \mathbf{a}^{0}_{k},   \hat{\mathbf{p}}_{k+1} \leq \mathbf{d}^{0}_{k},
                    \end{aligned}
                    \\[15pt]
                    \mathbf{g}_{k}^{0} + (f_{k}^{0})^{2}\mathbf{p}_{k+1}^{0},
                 & \text{otherwise},
    		\end{cases} 
    	\end{align}
        where $\hat{\mathbf{p}}_{k+1}:= (f_{k}^{1})^{2}\mathbf{p}_{k+1}^{1} - (f_{k}^{0})^{2}\mathbf{p}_{k+1}^{0}$.

        \underline{Case ii) - $\alpha_{k} = \{1,2,\dots,N-1\}$}

        \begin{align}
         \begin{split}\label{eq:def_TP_alpha_k}
                 y^{\alpha_{k}*}_{k}  = \begin{cases}
                    \begin{bmatrix} \hphantom{000}1\hphantom{000} & \hphantom{000}0\hphantom{000} 
                \end{bmatrix}^{\tp}, & \text{if } \ 
                    \begin{aligned}
                        & \hphantom{-.}\tilde{\mathbf{p}}_{k+1}^{\alpha_{k}} < \mathbf{d}_{k}^{\alpha_{k}},   -\check{\mathbf{p}}_{k+1}^{\alpha_{k}} < \mathbf{d}_{k}^{\alpha_{k}},
                    \end{aligned}
                    \\[15pt]
                    \begin{bmatrix} \hphantom{000}0\hphantom{000} & \hphantom{000}1 \hphantom{000}
                \end{bmatrix}^{\tp}, & \text{else if }  
                    \begin{aligned}
                        & \hphantom{-.}\tilde{\mathbf{p}}_{k+1}^{\alpha_{k}} > \mathbf{d}_{k}^{\alpha_{k}}, -\check{\mathbf{p}}_{k+1}^{\alpha_{k}} > \mathbf{d}_{k}^{\alpha_{k}},
                    \end{aligned}\\[15pt]
                \begin{bmatrix} \dfrac{\tilde{\mathbf{p}}_{k+1}^{\alpha_{k}}-\mathbf{a}_{k}^{\alpha_{k}}}{\tilde{\mathbf{p}}_{k+1}^{\alpha_{k}} + \check{\mathbf{p}}_{k+1}^{\alpha_{k}}} &  \dfrac{\check{\mathbf{p}}_{k+1}^{\alpha_{k}} + \mathbf{a}_{k}^{\alpha_{k}}}{\tilde{\mathbf{p}}_{k+1}^{\alpha_{k}} + \check{\mathbf{p}}_{k+1}^{\alpha_{k}}} \\ 
                \end{bmatrix}^{\tp}, & \text{otherwise } 
                \end{cases} 
        \end{split}\\
        \begin{split}\label{eq:adv_TP_alpha_k}
                 z^{0*}_{k}  = \begin{cases}
                    \begin{bmatrix}  \hphantom{000}1\hphantom{000} & \hphantom{000}0\hphantom{000}
                \end{bmatrix}^{\tp}, & \text{if } \ 
                    \begin{aligned}
                        & -\tilde{\mathbf{p}}_{k+1}^{\alpha_{k}} < \mathbf{a}^{\alpha_{k}}_{k},   \check{\mathbf{p}}_{k+1}^{\alpha_{k}} < \mathbf{a}^{\alpha_{k}}_{k},
                    \end{aligned}
                    \\[15pt]
                    \begin{bmatrix} \hphantom{000}0\hphantom{000} & \hphantom{000}1\hphantom{000}
                \end{bmatrix}^{\tp}, & \text{if } \ 
                    \begin{aligned}
                        & -\tilde{\mathbf{p}}_{k+1}^{\alpha_{k}} > \mathbf{a}^{\alpha_{k}}_{k},  \check{\mathbf{p}}_{k+1}^{\alpha_{k}} > \mathbf{a}^{\alpha_{k}}_{k},
                    \end{aligned}
                    \\[15pt]
                \begin{bmatrix} \dfrac{\tilde{\mathbf{p}}_{k+1}^{\alpha_{k}} + \mathbf{d}_{k+1}^{\alpha_{k}}}{\tilde{\mathbf{p}}_{k+1}^{\alpha_{k}} + \check{\mathbf{p}}_{k+1}^{\alpha_{k}}} & \dfrac{\check{\mathbf{p}}_{k+1}^{\alpha_{k}} - \mathbf{d}_{k+1}^{\alpha_{k}}}{\tilde{\mathbf{p}}_{k+1}^{\alpha_{k}} + \check{\mathbf{p}}_{k+1}^{\alpha_{k}}}\\
                \end{bmatrix}^{\tp}, & \text{otherwise,}
                \end{cases} 
        \end{split}
        \end{align}
        and the saddle-point value parameter satisfies:
        \begin{align}\label{eq:Val_scalar_alpha_k}
            \mathbf{p}_{k}^{\alpha_{k}} = 
            \begin{cases}
                \begin{aligned}
    				 \mathbf{g}_{k}^{\alpha_{k}} + & (f_{k}^{\alpha_{k}})^{2}\mathbf{p}_{k+1}^{\alpha_{K}},
                \end{aligned} & \text{if } \ 
                    \begin{aligned}
                        & -\tilde{\mathbf{p}}_{k+1}^{\alpha_{k}} < \mathbf{a}_{k}^{\alpha_{k}}, \check{\mathbf{p}}_{k+1}^{\alpha_{k}} < \mathbf{a}_{k+1}^{\alpha_{k}}, \\[2pt]
                        & \hphantom{-.}\check{\mathbf{p}}_{k+1}^{\alpha_{k}} < \mathbf{d}_{k+1}^{\alpha_{k}}, 
                    \end{aligned}
                    \\[20pt]
                \begin{aligned}
                    &\mathbf{g}_{k}^{\alpha_{k}} + (f_{k}^{\underline{\alpha}})^{2}\mathbf{p}_{k+1}^{\underline{\alpha}} + \mathbf{d}_{k+1}^{\alpha_{k}},
                \end{aligned} & \text{if } \ 
                    \begin{aligned}
                        & -\tilde{\mathbf{p}}_{k+1}^{\alpha_{k}} < \mathbf{a}_{k}^{\alpha_{k}}, \check{\mathbf{p}}_{k+1}^{\alpha_{k}} < \mathbf{a}_{k+1}^{\alpha_{k}}, \\[2pt]
                        & \hphantom{-.}\check{\mathbf{p}}_{k+1}^{\alpha_{k}} > \mathbf{d}_{k+1}^{\alpha_{k}}, 
                    \end{aligned}
                    \\[20pt]
                \begin{aligned}
                    & \mathbf{g}_{k}^{\alpha_{k}} + (f_{k}^{\overline{\alpha}})^{2}\mathbf{p}_{k+1}^{\overline{\alpha}}- \mathbf{a}_{k+1}^{\alpha_{k}},
                \end{aligned} & \text{if } \ 
                    \begin{aligned}
                        & -\tilde{\mathbf{p}}_{k+1}^{\alpha_{k}} > \mathbf{a}_{k}^{\alpha_{k}}, \check{\mathbf{p}}_{k+1}^{\alpha_{k}} > \mathbf{a}_{k+1}^{\alpha_{k}}, \\[2pt]
                        & -\tilde{\mathbf{p}}_{k+1}^{\alpha_{k}} < \mathbf{d}_{k+1}^{\alpha_{k}}, 
                    \end{aligned}
                    \\[20pt]
                \begin{aligned}
                   & \mathbf{g}_{k}^{\alpha_{k}} + (f_{k}^{\alpha_{k}})^{2}\mathbf{p}_{k+1}^{\alpha_{k}}  - \mathbf{a}_{k+1}^{\alpha_{k}} + \mathbf{d}_{k+1}^{\alpha_{k}},
                \end{aligned} & \text{if } \ 
                    \begin{aligned}
                        & -\tilde{\mathbf{p}}_{k+1}^{\alpha_{k}} > \mathbf{a}_{k}^{\alpha_{k}}, \check{\mathbf{p}}_{k+1}^{\alpha_{k}} > \mathbf{a}_{k+1}^{\alpha_{k}}, \\[2pt]
                        & -\tilde{\mathbf{p}}_{k+1}^{\alpha_{k}} > \mathbf{d}_{k+1}^{\alpha_{k}}, 
                    \end{aligned}
                    \\[20pt]
                \begin{aligned}
    				& \mathbf{g}_{k}^{0} + \dfrac{(f_{k}^{\alpha_{k}})^{4}(\mathbf{p}_{k+1}^{\alpha_{k}})^2 + \mathbf{a}_{k}^{\alpha_{k}}\mathbf{d}_{k}^{\alpha_{k}}}{\tilde{\mathbf{p}}_{k+1}^{\alpha_{k}} + \check{\mathbf{p}}_{k+1}^{\alpha_{k}}} \\[5pt] 
                    & + \dfrac{\tilde{\mathbf{p}}_{k+1}^{\alpha_{k}}\mathbf{d}_{k}^{\alpha_{k}} - \check{\mathbf{p}}_{k+1}^{\alpha_{k}}\mathbf{a}_{k}^{\alpha_{k}}}{\tilde{\mathbf{p}}_{k+1}^{\alpha_{k}} + \check{\mathbf{p}}_{k+1}^{\alpha_{k}}} \\[5pt]
                    & - \dfrac{(f_{k}^{\underline{\alpha}})^{2}\mathbf{p}_{k+1}^{\underline{\alpha}}(f_{k}^{\overline{\alpha}})^{2}\mathbf{p}_{k+1}^{\overline{\alpha}}}{\tilde{\mathbf{p}}_{k+1}^{\alpha_{k}} + \check{\mathbf{p}}_{k+1}^{\alpha_{k}}},
                \end{aligned} & \text{otherwise, }
    		\end{cases} 
    	\end{align}
        where
        \begin{gather*}
            \tilde{\mathbf{p}}_{k+1}^{\alpha_{k}} := (f_{k}^{\alpha_{k}})^{2}\mathbf{p}^{\alpha_{k}}_{k+1} - (f_{k}^{\underline{\alpha}})^{2}\mathbf{p}^{\underline{\alpha}}_{k+1},            \check{\mathbf{p}}_{k+1}^{\alpha_{k}} := (f_{k}^{\alpha_{k}})^{2}\mathbf{p}^{\alpha_{k}}_{k+1} - (f_{k}^{\overline{\alpha}})^{2}\mathbf{p}^{\overline{\alpha}}_{k+1}.
        \end{gather*}
        
    \underline{Case iii) - $\alpha_{k} = N$}
      \begin{align}
         \begin{split}\label{eq:def_TP_alpha_N}
                 y^{N*}_{k}  = \begin{cases}
                    \begin{bmatrix} 1 - \dfrac{\mathbf{a}^{N}_{k}}{\bar{\mathbf{p}}_{k+1}} 
 & \dfrac{\mathbf{a}^{N}_{k}}{\bar{\mathbf{p}}_{k+1}} 
                \end{bmatrix}^{\tp}, & \text{if } \ 
                    \begin{aligned}
                        & \bar{\mathbf{p}}_{k+1} > \mathbf{a}^{N}_{k},  \bar{\mathbf{p}}_{k+1} > \mathbf{d}^{N}_{k},
                    \end{aligned}
                    \\[15pt]
                    \begin{bmatrix} \hphantom{00} 0 \hphantom{0} & \hphantom{0000} 1 \hphantom{000}
                \end{bmatrix}^{\tp}, & \text{if } \ 
                    \begin{aligned}
                        & \bar{\mathbf{p}}_{k+1} \leq \mathbf{a}^{N}_{k},  \bar{\mathbf{p}}_{k+1} > \mathbf{d}^{N}_{k},
                    \end{aligned}
                    \\[15pt]
                    \begin{bmatrix} \hphantom{00} 1 \hphantom{0} & \hphantom{0000} 0 \hphantom{000}
                \end{bmatrix}^{\tp}, & \text{otherwise,}
                \end{cases} 
        \end{split}\\
        \begin{split}\label{eq:adv_TP_alpha_N}
                 z^{N*}_{k}  = \begin{cases}
                    \begin{bmatrix} \dfrac{\mathbf{d}^{N}_{k}}{\bar{\mathbf{p}}_{k+1}} & 1 - \dfrac{\mathbf{d}^{N}_{k}}{\bar{\mathbf{p}}_{k+1}} 
                \end{bmatrix}^{\tp}, & \text{if } \ 
                    \begin{aligned}
                        & \hat{\mathbf{p}}_{k+1} > \mathbf{a}^{N}_{k},  \hat{\mathbf{p}}_{k+1} > \mathbf{d}^{N}_{k},
                    \end{aligned}
                    \\[15pt]
                \begin{bmatrix} \hphantom{00} 1 \hphantom{00} & \hphantom{000} 0 \hphantom{000}
                \end{bmatrix}^{\tp}, & \text{otherwise,}
                \end{cases} 
        \end{split}
        \end{align}
       
        and the saddle-point value parameter is given by:
        \begin{align}\label{eq:Val_scalar_alpha_N}
            \mathbf{p}_{k}^{N} = 
            \begin{cases}
                \begin{aligned}
    				 \mathbf{g}_{k}^{N} + & (f_{k}^{N})^{2}\mathbf{p}_{k+1}^{N} - \mathbf{d}_{k}^{N} + \dfrac{\mathbf{a}^{0}_{k}\mathbf{d}^{0}_{k}}{\hat{\mathbf{p}}_{k+1}},
                \end{aligned} & \text{if } \ 
                    \begin{aligned}
                        & \hat{\mathbf{p}}_{k+1} > \mathbf{a}^{N}_{k},  \hat{\mathbf{p}}_{k+1} > \mathbf{d}^{N}_{k},
                    \end{aligned}
                    \\[15pt]
                \begin{aligned}
    				\mathbf{g}_{k}^{N} + & (f_{k}^{N-1})^{2}\mathbf{p}_{k+1}^{N-1} + \mathbf{d}_{k}^{N},
                \end{aligned} & \text{if } \ 
                    \begin{aligned}
                        & \hat{\mathbf{p}}_{k+1} > \mathbf{a}^{N}_{k},  \hat{\mathbf{p}}_{k+1} \leq \mathbf{d}^{N}_{k},
                    \end{aligned}
                    \\[15pt]
                    \mathbf{g}_{k}^{N} + (f_{k}^{N})^{2}\mathbf{p}_{k+1}^{N},
                 & \text{otherwise},
    		\end{cases} 
    	\end{align}
        where $\bar{\mathbf{p}}_{k+1}:= (f_{k}^{N})^{2}\mathbf{p}_{k+1}^{N} - (f_{k}^{N-1})^{2}\mathbf{p}_{k+1}^{N-1}$.
  
    The boundary condition of the saddle-point value recursion~\eqref{eq:Val_scalar_alpha_0},~\eqref{eq:Val_scalar_alpha_k},~\eqref{eq:Val_scalar_alpha_N} at $k = L+1$ is given by:
    \begin{equation}\label{eq:b_cond_val_dual}
        \mathbf{p}_{L+1}^{\alpha_{L+1}} := \mathbf{g}_{L+1}^{\alpha_{L+1}}, \forall \alpha_{L+1} \in V.    
    \end{equation}
    \frqed
\end{theorem}

\begin{proof}
    \underline{Case i) $\alpha_{k} = 0$}
    
    We begin by determining the NE takeover in both pure and mixed strategies of the \texttt{FlipDyn} state $\alpha_{k} = 0$. In the \texttt{FlipDyn} state $\alpha_{k} = 0$, for a $2 \times 2$ cost-to-go matrix, there are two pure and a mixed NE. We start by identifying the pure strategy NE.
    
    a) \textit{Pure NE strategy - Both the defender and adversary choose to play idle}. We begin by determining the condition when the defender always plays the action of remaining in the same node or idle. Given the quadratic takeover costs~\eqref{eq:ST_MV_cost_Q}, we compare the row entries of $\Xi_{k+1}^{0}$. First, we compare the row entries corresponding to the first column, representing the adversary choosing to remain idle: 
    \begin{equation*}
        \begin{aligned}
            (f_{k}^{0})^{2}\mathbf{p}_{k+1}^{0} \leq (f_{k}^{0})^{2}\mathbf{p}_{k+1}^{0} + \mathbf{d}_{k}^{0}.
        \end{aligned}
    \end{equation*}
    Next, comparing the row entries corresponding to the adversary choosing to takeover, provides us:
    \begin{equation*}
        \begin{aligned}
            &(f_{k}^{1})^{2}\mathbf{p}_{k+1}^{1} - \mathbf{a}_{k}^{0} \leq (f_{k}^{0})^{2}\mathbf{p}_{k+1}^{0} + \mathbf{d}_{k}^{0} - \mathbf{a}_{k}^{0}, \\ 
            \Rightarrow  & \underbrace{(f_{k}^{1})^{2}\mathbf{p}_{k+1}^{1} -  (f_{k}^{0})^{2}\mathbf{p}_{k+1}^{0}}_{\hat{\mathbf{p}}_{k+1}} \leq \mathbf{d}_{k}^{0}.
        \end{aligned}
    \end{equation*}
    Now, we will determine the condition under which the adversary opts to play idle, by comparing the column entries of $\Xi_{k+1}^{0}$. Comparing the column entries corresponding to the defender takeover action, provides us:
    \begin{equation*}
        \begin{aligned}
            (f_{k}^{0})^{2}\mathbf{p}_{k+1}^{0} \geq (f_{k}^{1})^{2}\mathbf{p}_{k+1}^{1} - \mathbf{a}_{k}^{0}.
        \end{aligned}
    \end{equation*}
    Finally, comparing the column entries corresponding to a defender takeover action, leads to:
    \begin{equation*}
        \begin{aligned}
            &(f_{k}^{0})^{2}\mathbf{p}_{k+1}^{0} + \mathbf{d}_{k}^{0} \geq (f_{k}^{0})^{2}\mathbf{p}_{k+1}^{0} + \mathbf{d}_{k}^{0} - \mathbf{a}_{k}^{0}, \\
            \Rightarrow & 0 \geq - \mathbf{a}_{k}^{0}.
        \end{aligned}
    \end{equation*}
    For quadratic costs, the condition $ 0 \geq - \mathbf{a}_{k}^{0}$ always holds. The cost-to-go value corresponding to the pure strategy of playing idle by both players corresponds to the entry $\Xi_{k+1}^{0}(1,1),$ where $X(i,j)$ corresponds to the $i^{\text{th}}$ row and $j^{\text{th}}$ column of the matrix $X$. Therefore, the saddle-point value under such a pure NE strategy is:
    \begin{equation*}
        \mathbf{p}_{k}^{0} = \mathbf{g}_{k}^{0} + (f_{k}^{0})^{2}\mathbf{p}_{k+1}^{0}. 
    \end{equation*}

    b) \textit{Pure NE strategy - The defender chooses an idle action while the adversary chooses to takeover}. Similar to the pure strategy of staying idle stated earlier, we can compare the entries of the matrix $\Xi_{k+1}^{0}$ to determine the conditions under which we will obtain the desired pure NE strategy. We leave out the comparisons and conditions for brevity. 

    c) \textit{Mixed NE strategy - Mixed strategies are played by both players} if none of the conditions corresponding to the pure NE strategy conditions are satisfied. We obtain a mixed NE strategy if:
    \begin{equation*}
        \begin{aligned}
            & (f_{k}^{1})^{2}\mathbf{p}_{k+1}^{1} - (f_{k}^{0})^{2}\mathbf{p}_{k+1}^{0} >  \mathbf{a}_{k}^{0}, \\ 
            & (f_{k}^{1})^{2}\mathbf{p}_{k+1}^{1} - (f_{k}^{0})^{2}\mathbf{p}_{k+1}^{0} >  \mathbf{d}_{k}^{0}.
        \end{aligned}
    \end{equation*}
    A mixed strategy NE takeover for any $2 \times 2$ game (cf.~\cite{shapley1952basic}) is given by: 
    \[
    y^{0*}_{k}  = \begin{bmatrix} \dfrac{\mathbf{a}_k(x)}{\hat{\mathbf{p}}_{k+1}} & 1 - \dfrac{\mathbf{a}_k(x)}{\hat{\mathbf{p}}_{k+1}}
            \end{bmatrix}^{\tp}, z^{0*}_{k}  = \begin{bmatrix} 1 - \dfrac{\mathbf{d}_k(x)}{\hat{\mathbf{p}}_{k+1}} & \dfrac{\mathbf{d}_k(x)}{\hat{\mathbf{p}}_{k+1}}
            \end{bmatrix}^{\tp}.
    \]
    The mixed saddle-point value (cf.~\cite{shapley1952basic}) of the $2 \times 2$ zero-sum matrix $\Xi_{k+1}^{0}$ is given by: 
    \begin{equation*}
        y_k^{0*^{\tp}}\Xi_{k+1}^{0}z_{k}^{0*} := (f_{k}^{0})^{2}\mathbf{p}_{k+1}^{0} + \mathbf{d}_{k}^{0} - \dfrac{\mathbf{a}_{k}\mathbf{d}_{k}}{\hat{\mathbf{p}}_{k+1}}.
    \end{equation*}
    Collecting all the saddle-point values and NE strategies corresponding to the pure and mixed strategy NE, we completely characterize the game for the node $\alpha_{k} = 0$ over the horizon of $L$ by~\eqref{eq:def_TP_alpha_0}, ~\eqref{eq:adv_TP_alpha_0}, and~\eqref{eq:Val_scalar_alpha_0}.

    \underline{Case ii) $\alpha_{k} = \{1,2,\dots,N-1\}$}

    For the case corresponding to the node $\alpha_{k} = \{1,2,\dots,N-1\}$ there are up to four pure NE strategies and a mixed NE strategy. Similar to the case of $\alpha_{k} = 0$, we start by identifying the pure NE strategies. We will derive the conditions for just a single pure NE strategy, and leave out the rest for brevity. 

    a) \textit{Pure NE strategy - Both the players opt to play idle.} In here, we compare the entries of the cost-to-go matrix $\Xi_{k+1}^{\alpha_{k}}$ to determine the condition for the pure NE strategy. We compare the row entries corresponding to the action of choosing the node $\alpha_{k}$ by the adversary, to obtain:
    \begin{equation*}
        \begin{aligned}
            & (f_{k}^{\alpha_{k}})^{2}\mathbf{p}_{k+1}^{\alpha_{k}} < (f_{k}^{\underline{\alpha}})^{2}\mathbf{p}_{k+1}^{\underline{\alpha}} + \mathbf{d}_{k}^{\alpha_{k}}, \\
           \Leftrightarrow & \underbrace{(f_{k}^{\alpha_{k}})^{2}\mathbf{p}_{k+1}^{\alpha_{k}} - (f_{k}^{\underline{\alpha}})^{2}\mathbf{p}_{k+1}^{\underline{\alpha}}}_{\tilde{\mathbf{p}}_{k+1}^{\alpha_{k}}} \leq \mathbf{d}_{k}^{\alpha_{k}}
        \end{aligned}
    \end{equation*}
    Likewise, we compare the row entries corresponding to the adversary choosing the node $\overline{\alpha}$ to obtain:
    \begin{equation*}
        \begin{aligned}
            & (f_{k}^{\overline{\alpha}})^{2}\mathbf{p}_{k+1}^{\overline{\alpha}} - \mathbf{a}_{k}^{\alpha_{k}} \leq (f_{k}^{\alpha_{k}})^{2}\mathbf{p}_{k+1}^{\alpha_{k}} + \mathbf{d}_{k}^{\alpha_{k}} - \mathbf{a}_{k}^{\alpha_{k}}, \\
            & \underbrace{(f_{k}^{\overline{\alpha}})^{2}\mathbf{p}_{k+1}^{\overline{\alpha}}- (f_{k}^{\alpha_{k}})^{2}\mathbf{p}_{k+1}^{\alpha_{k}}}_{-\check{\mathbf{p}}_{k+1}^{\alpha_{k}}} \leq \mathbf{d}_{k}^{\alpha_{k}}.
        \end{aligned}
    \end{equation*}
    The cost-to-go value corresponding to the pure NE strategy of choosing the node $\alpha_{k}$ by both the players, correspond to the entry $\Xi_{k+1}^{\alpha_{k}}(1,1)$. The corresponding saddle-point value parameter is given by:
    \begin{equation*}
        \mathbf{p}_{k}^{\alpha_{k}} = \mathbf{g}_{k}^{\alpha_{k}} + (f_{k}^{\alpha_{k}})^{2}\mathbf{p}_{k+1}^{\alpha_{k}}.
    \end{equation*}

    b) \textit{Pure NE strategies} - The remaining three pure strategy NE are:
    \begin{itemize}
        \item The defender chooses the node $\underline{\alpha}$ while the adversary opts the node $\alpha_{k}$.
        \item The adversary chooses the node $\overline{\alpha}$ while the defender opts the node $\alpha_{k}$.
        \item The defender opts to play the node $\underline{\alpha}$, whereas, the adversary opts to play the node $\overline{\alpha}$.
    \end{itemize}
    We leave out the derivation corresponding to the conditions required to obtain any of the stated NE strategy and saddle-point value parameter for brevity.

    c) \textit{Mixed NE strategies} - As shown for the case i), a mixed strategy NE is played when none of the conditions corresponding to the pure strategy NE are met. A mixed strategy NE takeover corresponding to the $2 \times 2$ cost-to-go matrix $\Xi_{k+1}^{\alpha_{k}}$ is given by:
    \begin{equation*}
            y_{k}^{\alpha_{k}*} = \begin{bmatrix} \dfrac{\tilde{\mathbf{p}}_{k+1}^{\alpha_{k}}-\mathbf{a}_{k}^{\alpha_{k}}}{\tilde{\mathbf{p}}_{k+1}^{\alpha_{k}} + \check{\mathbf{p}}_{k+1}^{\alpha_{k}}} \\[12pt] \dfrac{\check{\mathbf{p}}_{k+1}^{\alpha_{k}} + \mathbf{a}_{k}^{\alpha_{k}}}{\tilde{\mathbf{p}}_{k+1}^{\alpha_{k}} + \check{\mathbf{p}}_{k+1}^{\alpha_{k}}} \\ 
            \end{bmatrix}, 
            z_{k}^{\alpha_{k}*} = \begin{bmatrix} \dfrac{\tilde{\mathbf{p}}_{k+1}^{\alpha_{k}} + \mathbf{d}_{k+1}^{\alpha_{k}}}{\tilde{\mathbf{p}}_{k+1}^{\alpha_{k}} + \check{\mathbf{p}}_{k+1}^{\alpha_{k}}} \\[12pt] \dfrac{\check{\mathbf{p}}_{k+1}^{\alpha_{k}} - \mathbf{d}_{k+1}^{\alpha_{k}}}{\tilde{\mathbf{p}}_{k+1}^{\alpha_{k}} + \check{\mathbf{p}}_{k+1}^{\alpha_{k}}}\\ 
            \end{bmatrix}.
    \end{equation*}
    The mixed saddle-point value of the cost-to-go matrix $\Xi_{k+1}^{\alpha_{k}}$ is given by:
    \begin{equation*}
        \begin{aligned}
            y_k^{\alpha_{k}*^{\tp}} & \Xi_{k+1}^{\alpha_{k}}z_{k}^{\alpha_{k}*} := \dfrac{(f_{k}^{\alpha_{k}})^{4}(\mathbf{p}_{k+1}^{\alpha_{k}})^2 + \mathbf{a}_{k}^{\alpha_{k}}\mathbf{d}_{k}^{\alpha_{k}}}{\tilde{\mathbf{p}}_{k+1}^{\alpha_{k}} + \check{\mathbf{p}}_{k+1}^{\alpha_{k}}} \\[5pt] 
                    & + \dfrac{\tilde{\mathbf{p}}_{k+1}^{\alpha_{k}}\mathbf{d}_{k}^{\alpha_{k}} - \check{\mathbf{p}}_{k+1}^{\alpha_{k}}\mathbf{a}_{k}^{\alpha_{k}}}{\tilde{\mathbf{p}}_{k+1}^{\alpha_{k}} + \check{\mathbf{p}}_{k+1}^{\alpha_{k}}}  - \dfrac{(f_{k}^{\underline{\alpha}})^{2}\mathbf{p}_{k+1}^{\underline{\alpha}}(f_{k}^{\overline{\alpha}})^{2}\mathbf{p}_{k+1}^{\overline{\alpha}}}{\tilde{\mathbf{p}}_{k+1}^{\alpha_{k}} + \check{\mathbf{p}}_{k+1}^{\alpha_{k}}}.
        \end{aligned}
    \end{equation*}
    All the pure and mixed strategy NE and the corresponding saddle-point value parameters completely characterize the game for the nodes $\alpha_{k} = \{1,2,\dots,N-1\}$. The last condition in~\eqref{eq:def_TP_alpha_k},~\eqref{eq:adv_TP_alpha_k} and~\eqref{eq:Val_scalar_alpha_k} correspond to the mixed NE strategy and saddle-point value parameter. 

    \underline{Case iii) $\alpha_{k} = N$}

    We leave out the derivation of the pure and mixed NE strategies and corresponding saddle-point value parameter for $\alpha_{k} = N$, as they are analogous to the case $\alpha_{k} = 0$. 

    The boundary condition~\eqref{eq:b_cond_val_dual} corresponds to the state cost at time $k = L+1$, which is trivially  satisfied under the condition:
    \begin{equation*}
        \mathbf{p}_{L+1}^{\alpha_{L+1}} := \mathbf{g}_{L+1}^{\alpha_{L+1}}, \forall \alpha_{L+1} \in V.
    \end{equation*}
    \frQED
\end{proof}

Theorem~\ref{th:NE_Val_scalar_dual_deter} presents a closed-form solution for the \texttt{FlipDyn-G}~\eqref{eq:obj_def_E} game with NE takeover strategies independent of state for scalar linear dynamical systems. The dual deter model captures a specific structure of the general \texttt{FlipDyn-G} game. This structure enables us to complete the NE strategies and saddle-point value of the game in closed-form. The following remark indicates when the dual deter model maps to the \texttt{FlipDyn} model~\cite{FlipDyn_banik2022}.

\begin{remark}
    When the dual deter model consists of only two nodes, $\alpha = \{0,1\}$, the \texttt{FlipDyn-G} game reduces to a \texttt{FlipDyn}~\cite{FlipDyn_banik2022} model with a full state feedback control, with NE strategy and saddle-point value parameter as described in~\eqref{eq:def_TP_alpha_0},~\eqref{eq:adv_TP_alpha_0},~\eqref{eq:Val_scalar_alpha_0},~\eqref{eq:def_TP_alpha_N},~\eqref{eq:adv_TP_alpha_N}, and~\eqref{eq:Val_scalar_alpha_N}.
\end{remark}

Next, we illustrate the results of Lemma~\ref{lem:NE_Val_scalar} through two numerical examples.

\subsubsection*{\underline{Numerical Example I}} We evaluate the NE takeover strategy and saddle-point value of the \texttt{FlipDyn-G} game on an epidemic dynamic model, which is a discrete-time linear model capturing the dynamics of infection. Such a model can be mapped to a graph environment with four nodes, namely; susceptible, infected, recovered and deceased, termed as SIRD model. We assume the adversary as the underlying source of infection causing the transition between nodes. A government organization is represented as the defender preventing transitions to nodes, which can lead to significant losses. A typical epidemic model consists of fixed transition probabilities between different nodes. However, in this setup, the transition will be governed based on the NE takeover policies. A graphical representation of the SIRD model is shown in Figure~\ref{fig:Epidemic_model_graph}, with four \texttt{FlipDyn} states, susceptible (S), infected (I), recovered (R) and deceased (D). Therefore, the \texttt{FlipDyn} state can take on the value, $\alpha_{k} \in \{\text{S},\text{I},\text{R},\text{D}\}, \forall k \in \mathcal{K}.$

\begin{figure*}[ht]
	\begin{center}
		\subfloat[]{\includegraphics[width = 0.45\linewidth]{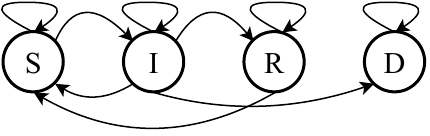}
			\label{fig:Epidemic_model_graph}	
		}
		\subfloat[]{\includegraphics[width = 0.45\linewidth]{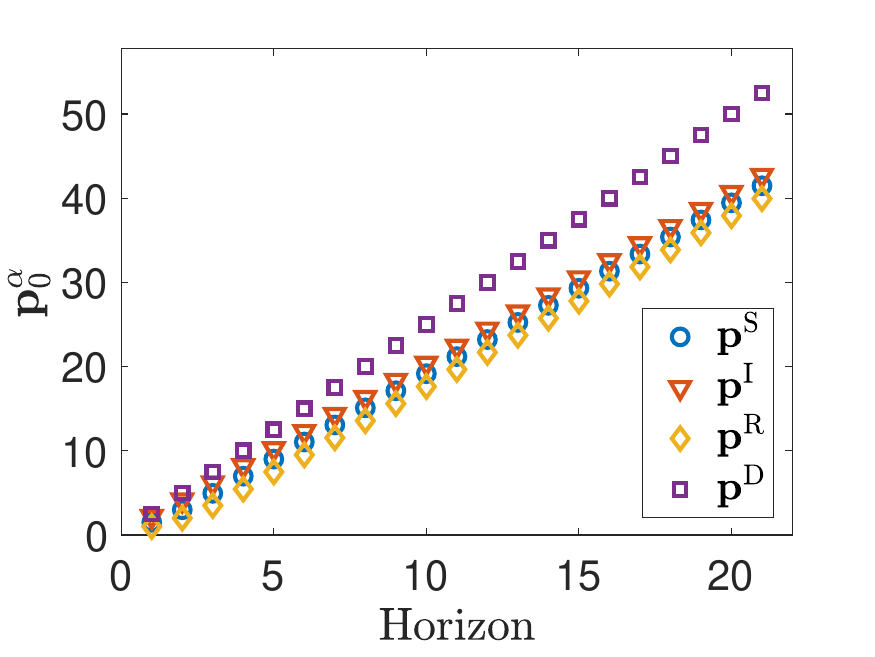}
			\label{fig:Val_EFDC}	
		}
    \caption{\small (a) An epidemic model represented as a graph with four nodes. The \texttt{FlipDyn} states of the graph are susceptible (S), Infected (I), Recovered (R), and Deceased (D). (b) Saddle-point parameters for each node $\alpha = \{\text{S,I,R,D}\}$, over time $k$, with horizon length $L = 20$.}	
		\label{fig:FlipDyn_TP_EPDG_a}
	\end{center}
\end{figure*}
This example presents only a \texttt{FlipDyn} dynamics, as the nodes do not have an underlying continuous state dynamics. In this example, we will consider the costs to be time-invariant, i.e., $\mathbf{g}^{\alpha}_{k} = \mathbf{g}^{\alpha}, \mathbf{d}^{\alpha}_{k} = \mathbf{d}^{\alpha}, $ and $\mathbf{a}_{k}^{\alpha} = \mathbf{a}^{\alpha}, \forall k \in \mathcal{K}$ and $\alpha \in \{\text{S,I,R,D}\}$. The state costs follow the order given by:
\begin{equation}
    \label{eq:state_costs_order_ED}
    \mathbf{g}^{\text{D}} > \mathbf{g}^{\text{I}} > \mathbf{g}^{\text{S}} > \mathbf{g}^{\text{R}}.
\end{equation}
The state costs~\eqref{eq:state_costs_order_ED} imply that the \texttt{FlipDyn} state of death ($\alpha = \text{D}$) has the highest cost, while the least is for the recovered ($\alpha = \text{R}$). Similarly, the defender and adversary takeover costs follow the order given by:
\begin{align}
    \label{eq:def_costs_order_ED}
    \mathbf{d}^{\text{R}} > \mathbf{d}^{\text{S}} > \mathbf{d}^{\text{I}} > \mathbf{d}^{\text{D}}, \\
    \mathbf{a}^{\text{D}} > \mathbf{a}^{\text{I}} > \mathbf{a}^{\text{S}} > \mathbf{a}^{\text{R}}.
\end{align}
The costs used in this numerical example are:
\begin{equation*}
    \begin{aligned}
        &\mathbf{g}^{\text{S}} = 1.5, \ \mathbf{g}^{\text{I}} = 2.2, \
        \mathbf{g}^{\text{R}} = 1.0, \
        \mathbf{g}^{\text{D}} = 2.5, 
        \\
        &\mathbf{d}^{\text{S}} = 0.7, \ \mathbf{d}^{\text{I}} = 0.5, \
        \mathbf{d}^{\text{R}} = 0.8, \
        \mathbf{d}^{\text{D}} = 0.2, 
        \\
        &\mathbf{a}^{\text{S}} = 0.5, \ \mathbf{a}^{\text{I}} = 0.7, \
        \mathbf{a}^{\text{R}} = 0.1, \
        \mathbf{a}^{\text{D}} = 0.9. \
    \end{aligned}    
\end{equation*}
We solve for the NE takeover strategies and saddle-point value using Lemma~\ref{lem:NE_Val_scalar}. Figure~\ref{fig:Val_EFDC} shows the saddle-point value parameters $\textbf{p}_{k}^{\alpha}, \alpha = \{\text{S,I,R,D}\}$ for a horizon length of $L = 20$. The saddle-point values corresponding to each of the nodes follow the order described in~\eqref{eq:state_costs_order_ED} indicating the cost in transitioning to the state $\alpha = D$ is the highest. We also observe that the value of the node $\alpha = I$ remains close to the other node states $\alpha = \{\text{R,D}\}$ reflective of the defender policy to prevent transition to $\alpha = \text{D}$. 

The defender and adversary policy for the state $\alpha = \text{I}$ are shown in Figures~\ref{fig:Def_pol_EFDG} and~\ref{fig:Adv_pol_EFDG}. Notice that the state $\alpha = \text{D}$ is a sink state, in other words, once you transition to it, you cannot transition to other states. We illustrate the policy on the state $\alpha = \text{I}$, as this state enables both the players to transition to any of the states. The defender's policy corresponds to transitioning to only the susceptible and recovered state, and no transition to the death state or remaining in the infected state. On the contrary, the adversary demonstrates high probability of transitioning to the death state and a low probability to the recovered state, while keeping the probability of transition to zero for the susceptible and infected state. 


\begin{figure*}[ht]
	\begin{center}
		\subfloat[]{\includegraphics[width = 0.45\linewidth]{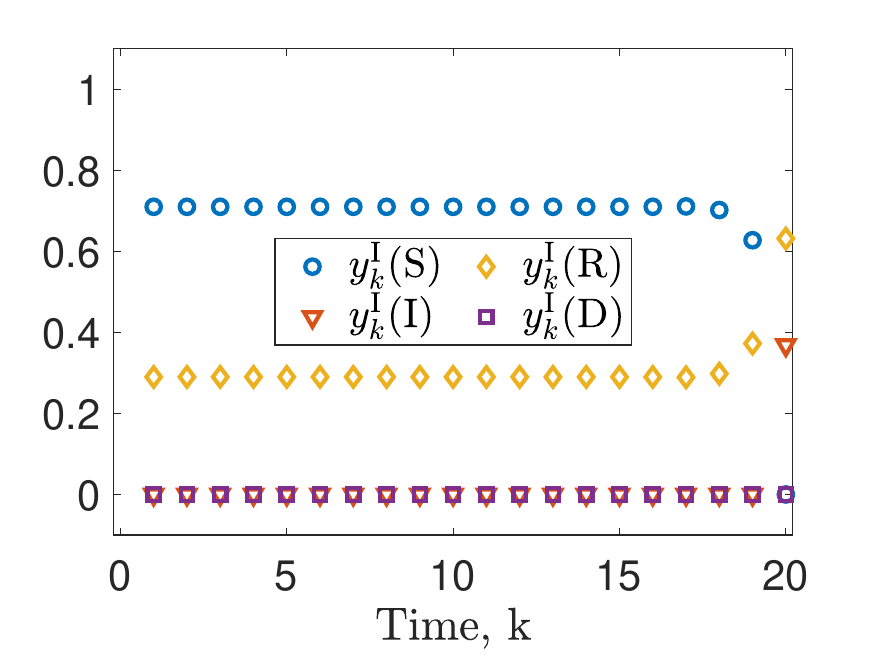}
			\label{fig:Def_pol_EFDG}	
		}
		\subfloat[]{\includegraphics[width = 0.45\linewidth]{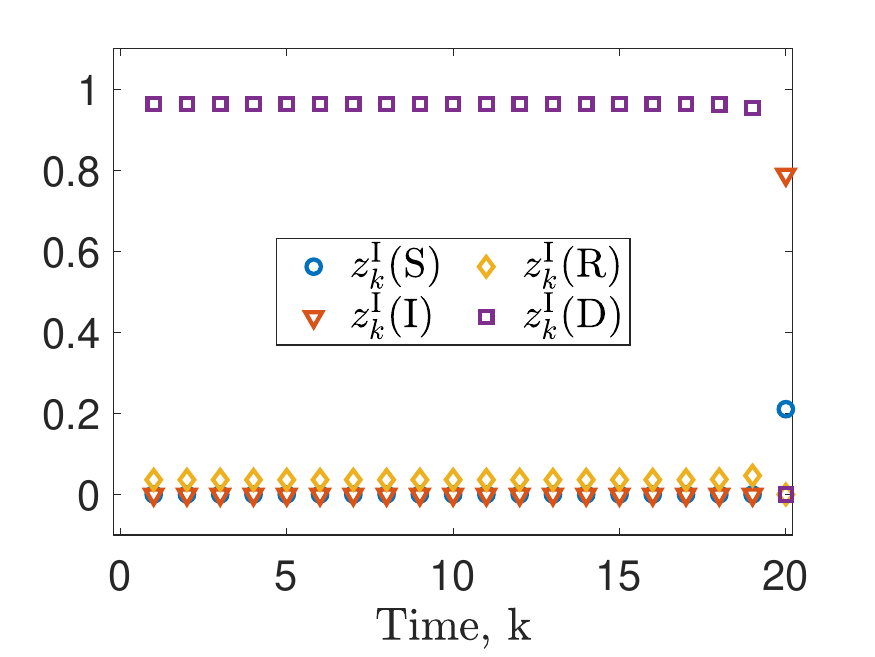}
			\label{fig:Adv_pol_EFDG}	
		}
    \caption{\small For the node $\alpha = \text{I}$, the NE policy of the (a) defender and (b) adversary, where $y_{k}^{\text{I}}(\alpha), z_{k}^{\text{I}}(\alpha), \alpha = \{\text{S,I,R,D}\}$ corresponds to the probability of selecting the takeover node $\alpha$, given $\alpha_k = \text{I}$.}	
		\label{fig:FlipDyn_TP_EPDG_b}
	\end{center}
\end{figure*}

\subsubsection*{\underline{Numerical Example II}} We evaluate the NE takeover strategy and saddle-point value of the \texttt{FlipDyn-G} game on a stock market Markov chain~\cite{peovski2022monitoring}, with node dynamics. Such a model consists of three nodes, namely; bull market, bear market, and stagnant market. An investor is represented as an adversary attempting to capitalize the market. A bull, bear, and stagnant market represents an increase, decrease and steady market growth, respectively. A graphical representation of such a stock market model is shown in Figure~\ref{fig:stock_market_model}, with three \texttt{FlipDyn} states, bull (Bu), Bear (Br), and stagnant (St). Therefore, the \texttt{FlipDyn} state can take on the value, $\alpha_{k} \in \{\text{Bu},\text{Br},\text{St}\}, \forall k \in \mathcal{K}.$
\begin{figure}[ht]
    \centering
    \includegraphics[width = 0.5\linewidth]{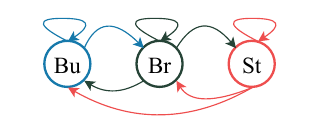}
    \caption{A stock market Markov chain model represented as a graph with three nodes. The \texttt{FlipDyn} states of the graph are Bull (Bu), Bear (Br), and Stagnant (St).}
    \label{fig:stock_market_model}
\end{figure}


\begin{figure*}[ht]
	\begin{center}
		\subfloat[]{\includegraphics[width = 0.45\linewidth]{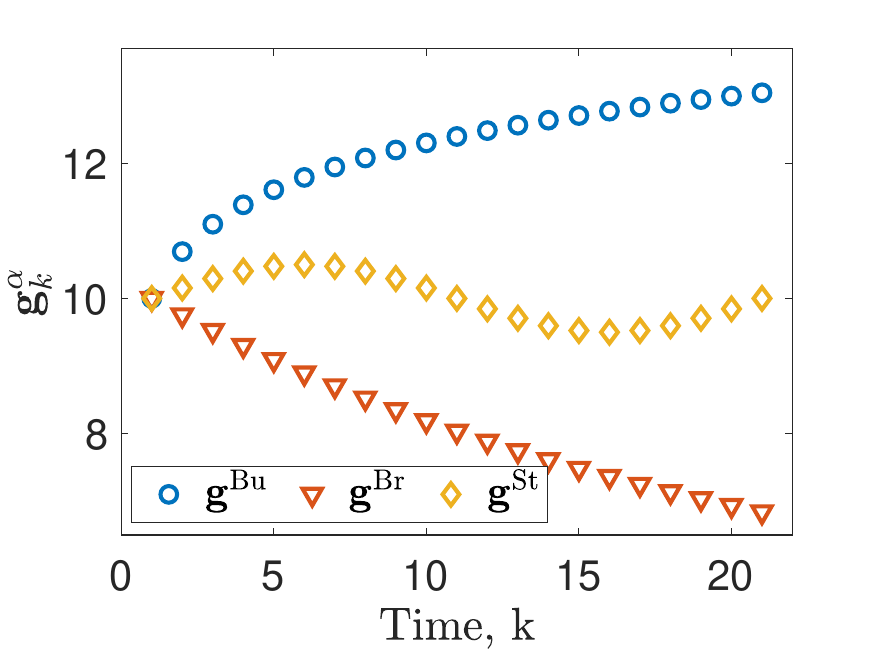}
			\label{fig:G_SMM}	
		}
		\subfloat[]{\includegraphics[width = 0.45\linewidth]{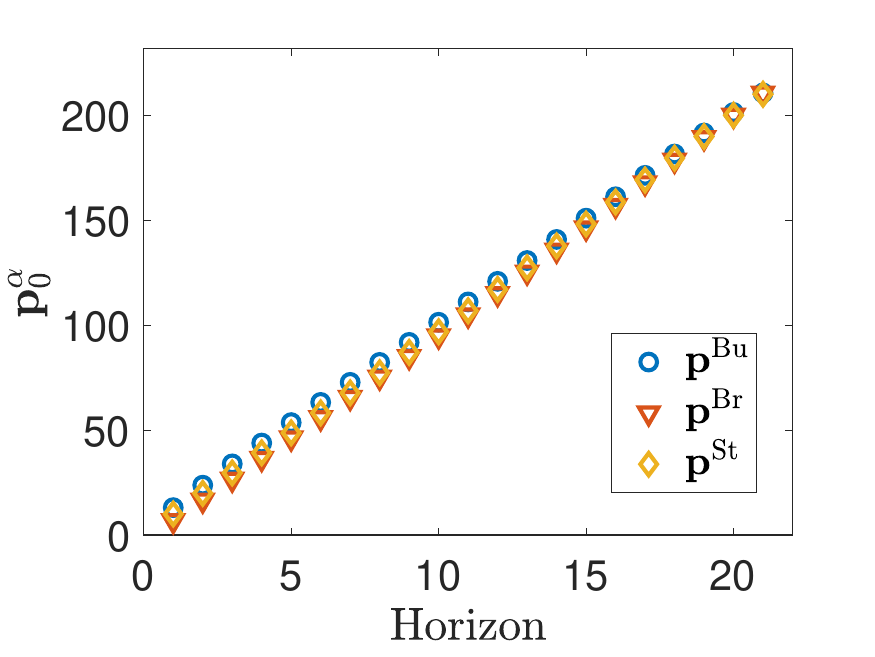}
			\label{fig:Val_SMM}	
		}
    \caption{\small (a) The state costs $\mathbf{g}_{k}^{\alpha}, \alpha \in \{\text{Bu, Br, St}\}$. (b) Saddle-point parameters for each node $\alpha = \{\text{Bu, Br, St}\}$, over time $k$, with horizon length $L = 20$.}	
		\label{fig:FlipDyn_TP_SMM_a}
	\end{center}
\end{figure*}

\begin{figure*}[ht]
	\begin{center}
		\subfloat[]{\includegraphics[width = 0.45\linewidth]{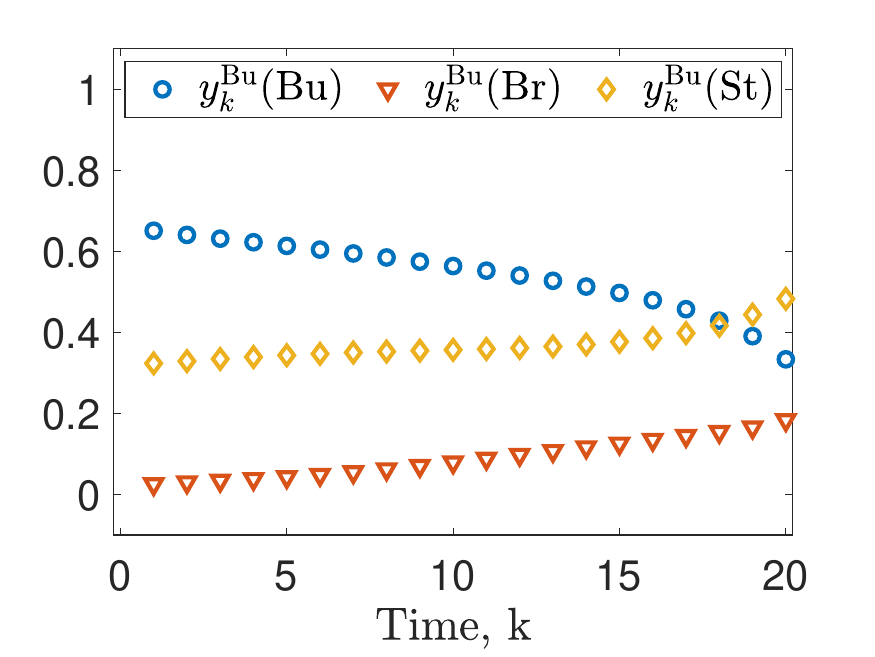}
			\label{fig:Def_pol_SMM_Bu}	
		}
            \subfloat[]{\includegraphics[width = 0.45\linewidth]{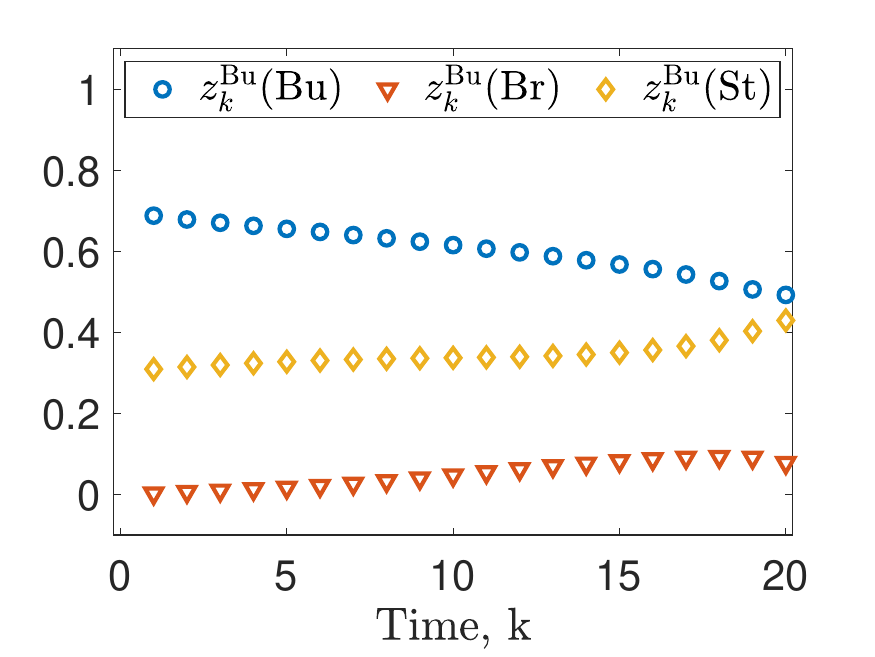}
			\label{fig:Adv_pol_SMM_Bu}	
		}
    \caption{\small  For the node $\alpha = \text{Br}$, the NE policy of the (c) defender and (d) adversary, where $y_{k}^{\text{Bu}}(\alpha), z_{k}^{\text{Bu}}(\alpha), \alpha = \{\text{Bu, Br, St}\}$ corresponds to the probability of selecting the takeover node $\alpha$, given $\alpha_k = \text{Bu}$.}	
		\label{fig:FlipDyn_TP_SMM_b}
	\end{center}
\end{figure*}

For this example, we will assume the costs and dynamics are time-invariant, i.e., $\mathbf{g}^{\alpha}_{k} = \mathbf{g}^{\alpha}, \mathbf{d}^{\alpha}_{k} = \mathbf{d}^{\alpha}, $ and $\mathbf{a}_{k}^{\alpha} = \mathbf{a}^{\alpha}, f_{k}^{\alpha_{k}} = f^{\alpha_{k}}, \forall k \in \mathcal{K}$ and $\alpha \in \{\text{Bu,Br,St}\}$. 
The state costs and node dynamics follow the order given by:
\begin{align}
    \label{eq:state_SMM}
    &\mathbf{g}^{\text{Bu}} > \mathbf{g}^{\text{Br}} > \mathbf{g}^{\text{St}}, \\ &
    f^{\text{Bu}} > f^{\text{Br}} > f^{\text{St}}. \label{eq:F_order_SMM}
\end{align}
The state costs~\eqref{eq:state_SMM} and dynamics~\eqref{eq:F_order_SMM} indicate the \texttt{FlipDyn} state of the bull market ($\alpha = \text{Bu}$) has the highest value with the least being the stagnant market ($\alpha = \text{St}$). Similarly, the defender and adversary takeover costs follow the order given by:
\begin{align}
    \label{eq:def_costs_order_SMM}
    \mathbf{d}^{\text{Bu}} > \mathbf{d}^{\text{St}} > \mathbf{d}^{\text{Br}}, \\
    \mathbf{a}^{\text{Br}} > \mathbf{a}^{\text{St}} > \mathbf{a}^{\text{Bu}}.
\end{align}
The dynamics and takeover costs used in this numerical example are:
\begin{equation*}
    \begin{aligned}
        &\mathbf{f}^{\text{Bu}} = 1.1, \ \mathbf{f}^{\text{Br}} = 0.95, \
        \mathbf{g}^{\text{St}} = 1.0, 
        \\
        &\mathbf{d}^{\text{Bu}} = 0.90, \ \mathbf{d}^{\text{Br}} = 0.50, \
        \mathbf{d}^{\text{St}} = 0.75, 
        \\
        &\mathbf{a}^{\text{Bu}} = 0.50, \ \mathbf{a}^{\text{Br}} = 0.90, \
        \mathbf{a}^{\text{St}} = 0.0.75,
    \end{aligned}    
\end{equation*}
The \texttt{FlipDyn} state costs are time-varying and indicate in Figure~\ref{fig:G_SMM}.
We solve for the NE takeover strategies and saddle-point value using Lemma~\ref{lem:NE_Val_scalar}. Figure~\ref{fig:Val_SMM} shows the saddle-point value parameters $\textbf{p}_{k}^{\alpha}, \alpha = \{\text{Bu,Br,St}\}$ for a horizon length of $L = 20$. At the start of the horizon, the difference between the saddle-point values follows the order~\eqref{eq:state_SMM}. However, as the horizon increases, the difference between saddle-point values between the \texttt{FlipDyn} states are indistinguishable, indicative of an efficient market (neither party gains). 

We only illustrate the defender and adversary policy for the state $\alpha = \text{Bu}$ shown in Figures~\ref{fig:Def_pol_SMM_Bu} and~\ref{fig:Adv_pol_SMM_Bu}, respectively. The policy trends of both players are quite similar, with a high probability of being in the bull market, followed by the stagnant market and bear market. The investor (adversary) indicates a higher probability of being in the bull market and maintains this probability throughout the time horizon. In contrast, the defender exhibits a relatively lower probability of being in the bull state, with the highest probability gradually shifting to transitioning to the stagnant state over time.

This numerical examples illustrates the use of the \texttt{FlipDyn} model in graphs in determining the node takeover strategies for each player. Additionally, it provides insight into the system's behavior for the given costs and the system's stability properties. These insights are useful while designing the costs, which further impact the takeover policies.

\section{Conclusion}\label{sec:Conclusion}
In this paper, we have introduced the \texttt{FlipDyn-G} framework, extending the \texttt{FlipDyn} model to a graph-based setting where each node represents a dynamical system. Our model captures the strategic interactions between a defender and an adversary who aim to control node state in a graph to minimize and maximize a finite horizon sum cost, respectively. The proposed framework results in a hybrid dynamical system in conjunction with continuous state node dynamics.

Our contributions include modeling and characterizing the \texttt{FlipDyn-G} game for general dynamical systems and deriving the corresponding Nash Equilibrium (NE) takeover strategies. Additionally, for scalar linear discrete-time dynamical systems with quadratic costs, we derived NE takeover strategies and saddle-point values that are independent of the continuous state of the system. For a finite state birth-death Markov chain, we derived analytical expressions for these NE strategies and values.

Through numerical studies involving epidemic models and linear dynamical systems with adversarial interactions, we have illustrated the applicability and effectiveness of our proposed methods. The results demonstrate that our approach can robustly determine optimal strategies for both players, enhancing the resilience and security of cyber-physical systems (CPS) against stealthy takeovers.

Future work will focus on extending this framework to more complex topologies and multi-agent systems, as well as exploring real-time applications in various CPS domains such as smart grids, autonomous vehicles, and industrial automation.

\begin{credits}
\subsubsection{\ackname} 
This research was supported in part by the NSF Award CNS-2134076 under the Secure and Trustworthy Cyberspace (SaTC) program and in part by the NSF CAREER Award ECCS-2236537.

\subsubsection{\discintname}
The authors have no competing interests to declare that are
relevant to the content of this article.

\end{credits}
%
%
%
\bibliographystyle{splncs04}
%
\bibliography{references}

\begin{thebibliography}{10}
\providecommand{\url}[1]{\texttt{#1}}
\providecommand{\urlprefix}{URL }
\providecommand{\doi}[1]{https://doi.org/#1}

\bibitem{acquaviva2017optimal}
Acquaviva, J., Mahon, M., Einfalt, B., LaPorta, T.: Optimal cyber-defense
  strategies for advanced persistent threats: a game theoretical analysis. In:
  2017 IEEE 36th Symposium on Reliable Distributed Systems (SRDS). pp.
  204--213. IEEE (2017)

\bibitem{FlipDyn_banik2022}
Banik, S., Bopardikar, S.D.: Flipdyn: A game of resource takeovers in dynamical
  systems. In: 2022 IEEE 61st Conference on Decision and Control (CDC). pp.
  2506--2511 (2022). \doi{10.1109/CDC51059.2022.9992387}

\bibitem{banik2023flipdyn}
Banik, S., Bopardikar, S.D.: Flipdyn with control: Resource takeover games with
  dynamics. arXiv preprint arXiv:2310.14484  (2023)

\bibitem{bowers2012defending}
Bowers, K.D., Van~Dijk, M., Griffin, R., Juels, A., Oprea, A., Rivest, R.L.,
  Triandopoulos, N.: Defending against the unknown enemy: Applying flipit to
  system security. In: International Conference on Decision and Game Theory for
  Security. pp. 248--263. Springer (2012)

\bibitem{bullo2009distributed}
Bullo, F., Cort{\'e}s, J., Martinez, S.: Distributed control of robotic
  networks: a mathematical approach to motion coordination algorithms, vol.~27.
  Princeton University Press (2009)

\bibitem{daian2020flash}
Daian, P., Goldfeder, S., Kell, T., Li, Y., Zhao, X., Bentov, I., Breidenbach,
  L., Juels, A.: Flash boys 2.0: Frontrunning in decentralized exchanges, miner
  extractable value, and consensus instability. In: 2020 IEEE symposium on
  security and privacy (SP). pp. 910--927. IEEE (2020)

\bibitem{ding2013stochastic}
Ding, J., Kamgarpour, M., Summers, S., Abate, A., Lygeros, J., Tomlin, C.: A
  stochastic games framework for verification and control of discrete time
  stochastic hybrid systems. Automatica  \textbf{49}(9),  2665--2674 (2013)

\bibitem{han2017mas}
Han, Y., Zhang, K., Li, H., Coelho, E.A.A., Guerrero, J.M.: Mas-based
  distributed coordinated control and optimization in microgrid and microgrid
  clusters: A comprehensive overview. IEEE Transactions on Power Electronics
  \textbf{33}(8),  6488--6508 (2017)

\bibitem{hespanha2017noncooperative}
Hespanha, J.P.: Noncooperative game theory: An introduction for engineers and
  computer scientists. Princeton University Press (2017)

\bibitem{johnson2015games}
Johnson, B., Laszka, A., Grossklags, J.: Games of timing for security in
  dynamic environments. In: Decision and Game Theory for Security: 6th
  International Conference, GameSec 2015, London, UK, November 4-5, 2015,
  Proceedings 6. pp. 57--73. Springer (2015)

\bibitem{kontouras2014adversary}
Kontouras, E., Tzes, A., Dritsas, L.: Adversary control strategies for
  discrete-time systems. In: 2014 European Control Conference (ECC). pp.
  2508--2513. IEEE (2014)

\bibitem{kontouras2015covert}
Kontouras, E., Tzes, A., Dritsas, L.: Covert attack on a discrete-time system
  with limited use of the available disruption resources. In: 2015 European
  Control Conference (ECC). pp. 812--817. IEEE (2015)

\bibitem{laszka2014flipthem}
Laszka, A., Horvath, G., Felegyhazi, M., Butty{\'a}n, L.: Flipthem: Modeling
  targeted attacks with flipit for multiple resources. In: Decision and Game
  Theory for Security: 5th International Conference, GameSec 2014, Los Angeles,
  CA, USA, November 6-7, 2014. Proceedings 5. pp. 175--194. Springer (2014)

\bibitem{lee2008cyber}
Lee, E.A.: Cyber physical systems: Design challenges. In: 2008 11th IEEE
  international symposium on object and component-oriented real-time
  distributed computing (ISORC). pp. 363--369. IEEE (2008)

\bibitem{leslie2015threshold}
Leslie, D., Sherfield, C., Smart, N.P.: Threshold flipthem: When the winner
  does not need to take all. In: Decision and Game Theory for Security: 6th
  International Conference, GameSec 2015, London, UK, November 4-5, 2015,
  Proceedings 6. pp. 74--92. Springer (2015)

\bibitem{leslie2017multi}
Leslie, D., Sherfield, C., Smart, N.P.: Multi-rate threshold flipthem. In:
  Computer Security--ESORICS 2017: 22nd European Symposium on Research in
  Computer Security, Oslo, Norway, September 11-15, 2017, Proceedings, Part II
  22. pp. 174--190. Springer (2017)

\bibitem{li1989overload}
Li, S.Q.: Overload control in a finite message storage buffer. IEEE
  Transactions on Communications  \textbf{37}(12),  1330--1338 (1989)

\bibitem{li2017dynamical}
Li, S.E., Zheng, Y., Li, K., Wu, Y., Hedrick, J.K., Gao, F., Zhang, H.:
  Dynamical modeling and distributed control of connected and automated
  vehicles: Challenges and opportunities. IEEE Intelligent Transportation
  Systems Magazine  \textbf{9}(3),  46--58 (2017)

\bibitem{merlevede2019time}
Merlevede, J., Johnson, B., Grossklags, J., Holvoet, T.: Time-dependent
  strategies in games of timing. In: Decision and Game Theory for Security:
  10th International Conference, GameSec 2019, Stockholm, Sweden, October
  30--November 1, 2019, Proceedings 10. pp. 310--330. Springer (2019)

\bibitem{miura2020modeling}
Miura, H., Kimura, T., Hirata, K.: Modeling of malware diffusion with the
  flipit game. In: 2020 IEEE International Conference on Consumer
  Electronics-Taiwan (ICCE-Taiwan). pp.~1--2. IEEE (2020)

\bibitem{mohan2020covert}
Mohan, A.M., Meskin, N., Mehrjerdi, H.: Covert attack in load frequency control
  of power systems. In: 2020 6th IEEE International Energy Conference
  (ENERGYCon). pp. 802--807. IEEE (2020)

\bibitem{moothedath2020game}
Moothedath, S., Sahabandu, D., Allen, J., Clark, A., Bushnell, L., Lee, W.,
  Poovendran, R.: A game-theoretic approach for dynamic information flow
  tracking to detect multistage advanced persistent threats. IEEE Transactions
  on Automatic Control  \textbf{65}(12),  5248--5263 (2020)

\bibitem{olfati2007consensus}
Olfati-Saber, R., Fax, J.A., Murray, R.M.: Consensus and cooperation in
  networked multi-agent systems. Proceedings of the IEEE  \textbf{95}(1),
  215--233 (2007)

\bibitem{peovski2022monitoring}
Peovski, F., Cvetkoska, V., Trpeski, P., Ivanovski, I.: Monitoring stock market
  returns: A stochastic approach. Croatian Operational Research Review
  \textbf{13}(1),  65--76 (2022)

\bibitem{rass2019cut}
Rass, S., K{\"o}nig, S., Panaousis, E.: Cut-the-rope: A game of stealthy
  intrusion. In: International Conference on Decision and Game Theory for
  Security. pp. 404--416. Springer (2019)

\bibitem{rass2023game}
Rass, S., K{\"o}nig, S., Wachter, J., Mayoral-Vilches, V., Panaousis, E.:
  Game-theoretic apt defense: An experimental study on robotics. Computers \&
  Security  \textbf{132},  103328 (2023)

\bibitem{saha2017flipnet}
Saha, S., Vullikanti, A., Halappanavar, M.: Flipnet: Modeling covert and
  persistent attacks on networked resources. In: 2017 IEEE 37th International
  Conference on Distributed Computing Systems (ICDCS). pp. 2444--2451. IEEE
  (2017)

\bibitem{shapley1952basic}
Shapley, L.S., Snow, R.: Basic solutions of discrete games. Contributions to
  the Theory of Games  \textbf{1},  27--35 (1952)

\bibitem{shi2011survey}
Shi, J., Wan, J., Yan, H., Suo, H.: A survey of cyber-physical systems. In:
  2011 international conference on wireless communications and signal
  processing (WCSP). pp.~1--6. IEEE (2011)

\bibitem{smith2015covert}
Smith, R.S.: Covert misappropriation of networked control systems: Presenting a
  feedback structure. IEEE Control Systems Magazine  \textbf{35}(1),  82--92
  (2015)

\bibitem{van2013flipit}
Van~Dijk, M., Juels, A., Oprea, A., Rivest, R.L.: Flipit: The game of
  “stealthy takeover”. Journal of Cryptology  \textbf{26},  655--713 (2013)

\bibitem{yao2023cheat}
Yao, Q., Xiong, X., Wang, Y.: Cheat-flipit: An approach to modeling and
  perception of a deceptive opponent. In: International Symposium on Dependable
  Software Engineering: Theories, Tools, and Applications. pp. 368--384.
  Springer (2023)

\bibitem{zhang2020defending}
Zhang, M., Zheng, Z., Shroff, N.B.: Defending against stealthy attacks on
  multiple nodes with limited resources: A game-theoretic analysis. IEEE
  Transactions on Control of Network Systems  \textbf{7}(4),  1665--1677 (2020)

\bibitem{Yang2015stability}
Zheng, Y., Li, S.E., Li, K., Wang, L.Y.: Stability margin improvement of
  vehicular platoon considering undirected topology and asymmetric control.
  IEEE Transactions on Control Systems Technology  \textbf{24}(4),  1253--1265
  (2015)

\bibitem{zheng2015stability}
Zheng, Y., Li, S.E., Wang, J., Cao, D., Li, K.: Stability and scalability of
  homogeneous vehicular platoon: Study on the influence of information flow
  topologies. IEEE Transactions on intelligent transportation systems
  \textbf{17}(1),  14--26 (2015)

\end{thebibliography}





\end{document}